%% file: main.tex
\documentclass[11pt]{article} 

\usepackage[top=1in, bottom=1in, left=1in, right=1in]{geometry}

\usepackage{setspace}
\usepackage{titlesec}
\titlespacing{\section}{0pt}{2pt}{2pt}
\titlespacing{\subsection}{0pt}{2pt}{2pt}
\everypar{\looseness=-1}

\usepackage{tikz}
\usetikzlibrary{arrows.meta,positioning, calc}

\usepackage[utf8]{inputenc} 
\usepackage[T1]{fontenc}    
\usepackage{hyperref}       
\usepackage{url}            
\usepackage{booktabs}       
\usepackage{amsfonts}       
\usepackage{amsthm, amsmath, amssymb, mathtools}
\usepackage{physics}
\usepackage{nicefrac}       
\usepackage{microtype}      
\usepackage{xcolor}         
\usepackage[shortlabels]{enumitem}

\usepackage{abstract}
\usepackage{titletoc} 

\newtheorem{theorem}{Theorem}[section]
\newtheorem{lemma}[theorem]{Lemma}
\newtheorem{definition}[theorem]{Definition}
\newtheorem{corollary}[theorem]{Corollary}
\newtheorem{proposition}[theorem]{Proposition}

\newtheorem*{remark}{Remark}

\newenvironment{construction}{
    \noindent\textbf{Construction:} 
}{
    \par 
}

\DeclarePairedDelimiter{\enc}{\langle}{\rangle}

\title{\rule{
\linewidth}{2.5pt}\\
\vspace{6pt}
\textbf{Realizable Circuit Complexity:}\\
\textbf{Embedding Computation in Space–Time}\\
\rule{\linewidth}{1pt}
}
\date{}

\author{
  Benjamin Prada\\
  Bellini College of AI,\\
  Cybersecurity, and Computing\\
  University of South Florida\\
  Tampa, FL 33617 \\
  \texttt{bprada@usf.edu} \\
  \and
  Ankur Mali\\
  Bellini College of AI,\\
  Cybersecurity, and Computing\\
  University of South Florida\\
  Tampa, FL 33617 \\
  \texttt{ankurarjunmali@usf.edu} \\
}

\begin{document}

\maketitle

\renewcommand{\abstractnamefont}{\normalfont\large\bfseries}
\input{sections/abstract}

\newpage

\startcontents[Main]
\printcontents[Main]{l}{1}{\section*{Table of Contents}}

\newpage

\input{sections/introduction}
\input{sections/background}
\input{sections/approach}
\input{sections/formalism}
\input{sections/thermo}
\input{sections/framework}
\input{sections/properties}
\input{sections/conclusion}

\stopcontents[Main]

\newpage

\startcontents[Appendix]
\printcontents[Appendix]{l}{1}{\section*{Appendix}}

\newpage

\setcounter{section}{0}
\renewcommand{\thesection}{\Alph{section}}
\renewcommand{\theHsection}{A\arabic{section}}

\input{sections/transformers}
\input{sections/extension}
\input{sections/history}

\newpage

\hrule
\bibliographystyle{ieeetr}
\bibliography{ref}

\end{document}

%% file: sections/abstract.tex
\begin{abstract}
\normalsize

Classical circuit complexity characterizes parallel computation in purely combinatorial terms, ignoring the physical constraints that govern real hardware. The standard classes $\mathbf{NC}$, $\mathbf{AC}$, and $\mathbf{TC}$ treat unlimited fan-in, free interconnection, and polynomial gate counts as feasible---assumptions that conflict with geometric, energetic, and thermodynamic realities. We introduce the family of \emph{realizable circuit classes} $\mathbf{RC}_d$, which model computation embedded in physical $d$-dimensional space. Each circuit in $\mathbf{RC}_d$ obeys conservative realizability laws: volume scales as $\mathcal{O}(t^d)$, cross-boundary information flux is bounded by $\mathcal{O}(t^{d-1})$ per unit time, and growth occurs through local, physically constructible edits. These bounds apply to all causal systems, classical or quantum. Within this framework, we show that algorithms with runtime $\omega(n^{d/(d-1)})$ cannot scale to inputs of maximal entropy, and that any $d$-dimensional parallel implementation offers at most a polynomial speed-up of degree $(d-1)$ over its optimal sequential counterpart. In the limit $d\to\infty$, $\mathbf{RC}_\infty(\mathrm{polylog})=\mathbf{NC}$, recovering classical parallelism as a non-physical idealization. By unifying geometry, causality, and information flow, $\mathbf{RC}_d$ extends circuit complexity into the physical domain, revealing universal scaling laws for computation.

\end{abstract}

%% file: sections/introduction.tex
\section{Introduction}

Circuit complexity has provided the foundation for understanding parallel computation, with classes such as $\mathbf{P/poly}$, $\mathbf{NC}$, $\mathbf{AC}$, and $\mathbf{TC}$ serving as canonical models.  
These classes, however, are defined in purely combinatorial terms: efficiency is measured by gate count, depth, and uniformity, while physical constraints---geometry, bandwidth, and thermodynamics---are ignored.  
As a result, they depend on assumptions that are physically unrealistic: unbounded fan-in, cost-free long-range connections, polynomial gate count of unbounded degree, and neglect of the heat generated by irreversible operations.  
Results from VLSI theory, Landauer’s principle, and Rent’s rule all indicate that such abstractions fail to capture the scaling limits of real computation.

We propose the framework of \emph{realizable circuits} $\mathbf{RC}_d$, a family of circuit complexity classes constrained by geometry and causality.  
An $\mathbf{RC}_d$ family represents circuits realizable in $d$-dimensional Euclidean space under a conservative set of physical laws:
\begin{enumerate}[(i)]
\item \textbf{Causality:} all components lie within a causal region $B(\mathbf{q_0},\mathcal{O}(t))$.
\item \textbf{Information Capacity:} information flux across the causal boundary is bounded as $\mathcal{O}(t^{d-1})$.
\item \textbf{Gate Quantization:} each gate occupies finite volume with bounded fan-in and fan-out.
\item \textbf{Local Uniformity:} circuit families must evolve through locally realizable extensions.
\end{enumerate}

These principles apply to all physically local systems---classical, quantum, or otherwise---making $\mathbf{RC}_d$ broadly applicable across physical models of computation.  

The novelty of $\mathbf{RC}_d$ is twofold. First, it yields scaling laws absent in conventional circuit complexity: algorithms requiring $\omega(n^{d/(d-1)})$ time cannot process a steady input stream, and any parallel implementation can achieve at most a polynomial speed-up of degree $(d-1)$ over its sequential form. Second, it refines existing complexity classes in physically meaningful terms: as $d \to \infty$, $\mathbf{RC}_\infty(\mathrm{polylog}) = \mathbf{NC}$, but for finite $d$, $\mathbf{RC}_d$ is strictly smaller, excluding families that depend on non-physical assumptions.

%% file: sections/background.tex
\section{Background}

This section establishes the terms and context prerequisite for understanding the work. We begin by introducing the standard classes of circuit complexity theory and their associated hierarchy. Afterward, we motivate the need for a new class characterizing the limits of real-world circuits.

\subsection{Circuit complexity}

Circuit complexity measures the intrinsic difficulty of computing a function by the size and depth of a Boolean circuit that realizes it.  
For an input of length $n$, the circuit size---the total number of gates---is required to grow at most polynomially in $n$, while the circuit depth bounds the degree of parallelism achievable.

The three principal circuit families are
\[
\mathbf{NC} := \bigcup_{k>0}\mathbf{NC}^k, \qquad
\mathbf{AC} := \bigcup_{k>0}\mathbf{AC}^k, \qquad
\mathbf{TC} := \bigcup_{k>0}\mathbf{TC}^k,
\]
related as
\[
\mathbf{NC}^k \subseteq \mathbf{AC}^k \subseteq \mathbf{TC}^k \subseteq \mathbf{NC}^{k+1}.
\]
Each class consists of functions computable by polynomial-size circuits of depth $\mathcal{O}(\mathrm{polylog}(n))$.

The distinction lies in the permitted gate sets.  
In $\mathbf{NC}$, gates have bounded fan-in (typically 2).  
$\mathbf{AC}$ extends this to unbounded fan-in Boolean gates, while $\mathbf{TC}$ further includes \emph{threshold gates} that output $1$ if and only if the number of $1$-inputs exceeds a fixed threshold.  
These classes serve as the canonical formalization of parallel computation, abstracting time, communication, and synchronization costs within a combinatorial framework.

\subsection{Uniform circuits}

By default, circuit complexity classes require only that a circuit of the specified size and depth \emph{exists} for each input length $n$.  
The construction of these circuits may vary arbitrarily with $n$—in particular, the sequence $\{C_n\}$ need not be computable by any efficient algorithm.

To address this, the notion of \emph{uniformity} restricts the circuit family to those that can be generated algorithmically.  
A function $f$ in a circuit class $C$ is said to be \emph{$T$-uniform} in $C$ if there exists a function $F$ in class $T$ such that, for all $n$, $F(n)$ outputs a $C$-circuit computing $f$ on inputs of length $n$.  
Uniformity thus ensures that the mapping from input size to circuit structure is itself computable within resource bound $T$.

Different works implicitly assume various levels of uniformity—for instance, $\mathbf{NC}$ is often taken to mean $\mathbf{DLOGTIME}$-uniform $\mathbf{NC}$.  
In this work, all such assumptions are stated explicitly: unless otherwise specified, the classes $\mathbf{NC}, \mathbf{AC}$, and $\mathbf{TC}$ refer to their non-uniform definitions, and uniform restrictions are introduced only when required.

\subsection{Impracticality}

The standard parallel circuit classes—$\mathbf{NC}$, $\mathbf{AC}$, and $\mathbf{TC}$—offer elegant abstractions for reasoning about Boolean circuit families.  
However, when evaluated against physical constraints, these models prove impractical.  
They assume that arbitrarily many components can be interconnected without penalty and treat polynomial gate count as a universal proxy for feasibility, regardless of degree.  
Such assumptions were convenient in the 1970s–1990s for studying theoretical parallelism, but subsequent work in VLSI complexity \cite{MeadConway1980VLSI,Thompson1979ATVLSI}, Rent’s rule \cite{LandmanRusso1971Rent,Lanzerotti2005RentIBMJRD}, and energy-dominance analysis \cite{Horowitz2014EnergyProblem,SzeNASEM2024DataMovement} have demonstrated that they are physically unrealistic.  
\begin{enumerate}
    \item \textbf{Polynomial gate count as a feasibility proxy.}  
    In $\mathbf{P/poly}$ and its subclasses, feasibility is defined by the existence of circuits with polynomially many gates.  
    However, the polynomial degree can mask exponential resource blow-up in area, wiring, or latency.  
    VLSI area–time tradeoffs \cite{Thompson1979ATVLSI} show that global operations such as sorting or FFT require $A \cdot T^2 = \Omega(N^2)$ resources, even when gate count is modest.  
    Geometry and interconnects dominate cost, making gate count alone a poor indicator of realizability.

    \item \textbf{Unbounded fan-in.}  
    The classes $\mathbf{AC}$ and $\mathbf{TC}$ permit unbounded fan-in gates, effectively assuming that arbitrarily many signals can be aggregated within a single layer.  
    Foundational VLSI results \cite{MeadConway1980VLSI} and empirical wiring laws \cite{LandmanRusso1971Rent,Lanzerotti2005RentIBMJRD,Stroobandt2010RentRuleSurvey} demonstrate that I/O pins and interconnects scale sub-linearly with logic, imposing strict upper bounds on fan-in and fan-out.  
    Thus, unbounded fan-in is not physically realizable.

    \item \textbf{Neglect of energy and data movement.}  
    Data movement dominates energy consumption in modern computing \cite{Horowitz2014EnergyProblem,SzeNASEM2024DataMovement}.  
    Because $\mathbf{NC/AC/TC}$ ignore interconnects and memory-transfer energy, they overestimate the scalability of shallow-depth circuits that depend on global aggregation or broadcasting.
\end{enumerate}

In summary, classical circuit classes illuminate logical parallelism but not physical feasibility.  
They omit the geometric, thermodynamic, and data-movement constraints that govern real hardware.  
We therefore introduce the $\mathbf{RC}_d$ family of \emph{realizable circuits}—a framework that redefines circuit complexity under explicit causal, spatial, and energetic limits.

\subsection{Prior work}\label{sec:prior}

Spatial and physical constraints on computation have received surprisingly little formal treatment in
complexity theory. Nearly all foundational work in this direction traces back to a concise three-page note
by Fisher~\cite{Fisher1988}. Fisher modeled a $d$-dimensional parallel machine as a finite-density grid of
processors communicating at bounded signal speed. He proved that any computation performed within time $t$
must lie inside a causal sphere of radius $r=\Theta(t)$ and, consequently, that the achievable execution time
for a computation with $I$ inputs, $K$ outputs, and total processor--time $T$ satisfies
\begin{equation}\label{eq:fisher}
t = \Omega\big(\max\{\,I^{1/d},\,K^{1/d},\,T^{1/(d+1)}\,\}\big),
\end{equation}
henceforth called \emph{Fisher’s bound}.  This inequality captures the geometric latency of communication and
establishes a lower limit on parallel speedup as a function of spatial dimension $d$.

Subsequent efforts, most notably the \emph{Spatial Machine} framework of Leighton, Rose, and
Saxe~\cite{spatial_machines} and related analyses~\cite{JEN1989331,WORLEY1992361}, extended Fisher’s
geometric intuition.  These models formalize computation as a grid of finite-state processors, each capable
of communicating only with local neighbors.  They confirm that the physical embedding of computation in
finite-dimensional space imposes fundamental communication and latency limits beyond those captured by
traditional circuit complexity.

\paragraph{Limitations of Fisher’s formulation.}
Fisher’s argument integrates communication and work over the entire runtime $t$, bounding the \emph{average}
information flux through the causal region.  However, this treatment does not constrain the
\emph{instantaneous} rate at which information or entropy can traverse the system boundary.
If a computation requires all $n$ input bits to be simultaneously present within the causal region—corresponding
to a single-shot, maximal-entropy evaluation—the relevant bound depends on the instantaneous surface area of
the causal sphere rather than its integrated volume.  Under this stricter physical interpretation, the
realizability constraints developed in our framework yield
\begin{equation}\label{eq:rc-stronger}
t = \Omega(n^{1/(d-1)}),
\end{equation}
which is a strictly tighter lower bound than~\eqref{eq:fisher} for $d>2$.

\paragraph{Beyond geometry: flux and thermodynamics.}
The $\mathbf{RC}_d$ formalism generalizes Fisher’s geometric reasoning in three key ways:
\begin{enumerate}
    \item \textbf{Flux-bounded computation.}  
    Rather than bounding total processor--time $T$, $\mathbf{RC}_d$ constrains the instantaneous entropy flux
    across the causal boundary to $\Phi(t)=\mathcal{O}(t^{d-1})$, enforcing local information conservation and
    directly leading to~\eqref{eq:rc-stronger}.
    \item \textbf{Thermodynamic grounding.}  
    Fisher’s model assumes reversible communication; our formulation incorporates Landauer’s principle, which
    limits the rate of irreversible erasures (and hence heat dissipation) by the same $\mathcal{O}(t^{d-1})$
    geometric throughput bound.
    \item \textbf{Generality across physical substrates.}  
    Because these constraints derive from geometry and conservation laws, they apply uniformly to both
    classical and quantum circuits, establishing $\mathbf{RC}_d$ as a universal upper bound on physically
    realizable parallelism.
\end{enumerate}

\paragraph{Toward a physically grounded complexity framework.}
Fisher’s inequality~\eqref{eq:fisher} provides the first geometric lower bound on parallel computation.
The $\mathbf{RC}_d$ framework recovers it as a special case—corresponding to the time-averaged regime—and
strengthens it by enforcing instantaneous flux, bounded fan-in, and thermodynamic consistency.
Crucially, by defining $\mathbf{RC}_d$ as a bona fide \emph{complexity class}, this work bridges physical
computation and nonuniform circuit theory, yielding a unified foundation for reasoning about the true limits
of scalable parallelism.

%% file: sections/approach.tex
\section{Physical Constraints}

We propose three fundamental bounds on circuit proportions, motivated by principles from physics and information theory. In Section~\ref{formalism}, we provide a formal derivation from a physical model of computation. (For a consideration of application to quantum systems, please refer to Appendix~\ref{sec:quantum}.)

\paragraph{Volume constraint.}
A circuit must reside within a causal region.  
Without loss of generality, model this region as an expanding ball $B(\mathbf{0},r)\subset\mathbb{R}^3$ with radius $r=\mathcal{O}(t)$.  
Its volume grows as $\mathcal{O}(t^3)$; hence, the number of gates in a circuit $C_n$ is bounded by $\mathcal{O}(t^3)$ under finite wire density.  
The following generalizes to $d$ dimensions.

\begin{proposition}\label{prop:vol}
Consider a circuit $C(t)$ evolving in $\mathbb{R}^d$.  
If the causal region containing it at time $t$ is a ball of radius $r=\mathcal{O}(t)$ and adjacent wires are separated by a minimum distance $\ell>0$, then
$
|C(t)| = \mathcal{O}(t^d).
$
\end{proposition}
\vspace{-11pt}

\paragraph{Surface constraint.}
The entropy entering or leaving a causal region is proportional to its surface area.  
Model the entropy change as a vector field in $\mathbb{R}^d$ with magnitude bounded by a maximum signal speed $c$.  
The total entropy flux across the boundary is therefore $\mathcal{O}(r^{d-1})=\mathcal{O}(t^{d-1})$, since the surface area of a $d$-sphere scales as $r^{d-1}$.  
Assuming incompressible flow, the same bound holds for non-spherical regions.

\begin{proposition}\label{prop:width}
Consider a circuit $C(t)$ evolving in $\mathbb{R}^d$.  
If its causal region at time $t$ has radius $r=\mathcal{O}(t)$ and the entropy flow is incompressible with speed at most $c$, then both entropy influx and outflux are bounded by $\mathcal{O}(t^{d-1})$.  
Hence the circuit width—its maximal entropy exchange across any cut—is $\mathcal{O}(t^{d-1})$.
\end{proposition}

\paragraph{Gate constraint.}
Computation within a bounded spatial region must perform finite work.  
We impose this by constraining the primitive gate set: each gate is a constant-size subcircuit obeying Proposition~\ref{prop:width}, and therefore performs only finite work.

\begin{proposition}\label{prop:gates}
Every circuit $C(t)$ can be realized as a network of gates of $\mathcal{O}(1)$ volume.  
By Proposition~\ref{prop:width}, each such gate has $\mathcal{O}(1)$ entropy influx and outflux.
\end{proposition}
\vspace{-0.6cm}
\paragraph{Discussion.}
These constraints are deliberately conservative: they rely solely on causal geometry and finite information density.  
They apply uniformly across all admissible gate families—classical or quantum—and are independent of algebraic or technological details.  
Thus, the bounds represent universal physical limits on circuit growth and throughput, forming the foundation of the $\mathbf{RC}_d$ framework.

%% file: sections/formalism.tex
\section{Formalism for Physical Constraints}\label{formalism}

Propositions~\ref{prop:vol}--\ref{prop:gates} establish three complementary limits: a volume bound ($\mathcal{O}(t^d)$ gates), a throughput bound ($\mathcal{O}(t^{d-1})$ width), and a density bound ($\mathcal{O}(1)$ channel capacity). All are rooted in physics: circuits must reside in a causal ball, interactions are mediated locally across its boundary, and components cannot be placed arbitrarily close together.

To formalize these constraints, we model physical computation as local Hamiltonian systems evolving inside causal regions. This ensures our abstraction assumes only unavoidable physical principles: finite signal speed, locality, and conservation. In particular, the $\mathcal{O}(t^d)$ volume bound follows from comparison geometry (Bishop--Gromov~\cite{Bishop1964GeometryOM}) under nonnegative Ricci curvature, while the $O(t^{d-1})$ flux bound follows from Liouville’s theorem on Hamiltonian flows.

\subsection{Volume bound}

\begin{definition}[Hamiltonian system]\label{def:hamilton}
A $d$-dimensional Hamiltonian system on $n$ particles is a real dynamical system $(X^n, \phi_t)$ such that:
\begin{itemize}
    \item The phase space is parameterized as $X^n=Q^{dn}\times P^{dn}$; i.e., a set of $n$ particles in $d$ dimensions, with real position coordinates $\mathbf{q} = (q_1,\dots,q_{dn})$ and momentum coordinates $\mathbf{p} = (p_1,\dots,p_{dn})$.
    
    \item There exists a scalar function $H(\mathbf{q},\mathbf{p}, t)$ called the \textbf{Hamiltonian} such that the flow $\phi_t$ obeys
    $$
    \mathbf{\dot{q}} = \frac{\partial H}{\partial \mathbf{p}}, \qquad \mathbf{\dot{p}} = -\frac{\partial H}{\partial \mathbf{q}}.
    $$
\end{itemize}

A Hamiltonian system is \textbf{local} if the evolution of each particle $(\mathbf{q}, \mathbf{p})\in X$ depends on particles only within a bounded neighborhood.
\end{definition}

\begin{definition}[Causal regions in space--time]\label{def:causal}
Fix $d>0$. Let $(\mathbf{q_0},t_0)\in\mathbb{R}^d\times\mathbb{R}$ and let $c>0$ denote a maximum signal speed. Write $d_g(\cdot, \cdot)$ for the geodesic distance induced by a Riemannian metric on $\mathbb{R}^d$. The (two-sided) \textbf{causal cone} at $(\mathbf{q_0},t_0)$ is
\[
\mathcal{C}(\mathbf{q_0},t_0)
\;:=\;
\Big\{(\mathbf{q},t)\in\mathbb{R}^d\times\mathbb{R}:\ d_g(\mathbf{q},\mathbf{q_0}) \le c\,|t-t_0|\Big\}.
\]
A slice of this causal cone at a fixed time $t$ is the \textbf{causal ball}
$$
\mathcal{C}_t(\mathbf{q_0}, t_0):=\{\mathbf{q}\in \mathbb{R}^d: d_g(\mathbf{q}-\mathbf{q_0})\le c\,|t-t_0|\}.
$$
\end{definition}

Notice in Definition~\ref{def:causal} that geodesic distance allows for curved spatial geometry. When space is flat, $d_g(\cdot,\cdot)$ is the usual Euclidean metric. When positive Ricci curvature is present, it follows from comparison geometry (Bishop--Gromov \cite{Bishop1964GeometryOM}) that the volume within a boundary is strictly smaller than in the Euclidean case.

Now we may define what it means to realize a circuit by a local Hamiltonian system.

\begin{definition}[Circuit realization]\label{def:realize}
For a circuit $C$ with binary wires $\{w_0, \ldots, w_N\}=W\in\mathbb{Z}_2^n$, let $w_i(t)$ denote the value of the $i$-th wire at discrete timestep $t\in \mathbb{N}$. A local Hamiltonian system $(X^n, \phi_t)$ is said to \textbf{realize} a circuit $C$ if there exists a set of states $\{v_0, \ldots v_N\}=V \subseteq X$ and decoding function $\mathrm{Dec}:V\to \mathbb{Z}_2$ such that $\mathrm{Dec}(v_i(k\,t))\mapsto w_i(t)$ for some propagation delay $k>0$.
\end{definition}

Our definition of circuit realization requires that the local Hamiltonian system has a set of states that can be directly interpreted as the intermediate values (wires) of the circuit computation. Each decoded value $\mathrm{Dec}(v_i)$ needs only coincide with the wire signal $w_i$ after time intervals of a fixed period $k>0$. This $k$ represents the propagation delay between layers plus any settling time, thus controlling the speed at which the system computes the target function. Note that the locality of the realization is implicit in the reliance on a local system.

The following theorem lays the foundation for proving Proposition~\ref{prop:vol}. A standard sufficient condition for $\mathcal{O}(r^d)$ growth is non-negative Ricci curvature.

\begin{theorem}[Causal volume bound]\label{thm:causal-volume}
Let $Q(t)\subset\mathbb{R}^d$ be a measurable spatial region at time $t\ge t_0$, and suppose
the spatial metric has non-negative Ricci curvature. If the region is causal in the sense that there exists $(\mathbf{q_0}, t_0)\in \mathbb{R}^d\times \mathbb{R}$ such that
\[
Q(t) \subseteq \mathcal{C}_t(\mathbf{q_0}, t_0),
\quad\forall t\ge t_0,
\]
then its spatial volume satisfies
\[
\mathrm{Vol}\big(Q(t)\big)
\le k\,c^d(t-t_0)^d
= \mathcal{O}\!\big((t-t_0)^d\big),
\]
for a constant $k>0$ depending only on the ambient geometry.
\end{theorem}

\begin{proof}
Given $Q(t) \subseteq \mathcal{C}_t(\mathbf{q_0}, t_0)$, it follows that 
$\mathrm{Vol}(Q(t))\le \mathrm{Vol}(\mathcal{C}_t(\mathbf{q_0}, t_0))$. Under non-negative Ricci curvature, Bishop--Gromov comparison yields
$\mathrm{Vol}(B(\mathbf{q_0}, r))\le k\,r^d$ for all $r>0$ and some $k>0$. Therefore,
\begin{align*}
\mathrm{Vol}(Q(t))&\le \mathrm{Vol}(\mathcal{C}_t(\mathbf{q_0}, t_0))\\
&=\mathrm{Vol}(B(\mathbf{q_0}, c(t - t_0)))\\
&\le k\,c^d(t-t_0)^d\\
&= \mathcal{O}((t-t_0)^d).
\end{align*}
\end{proof}

Proposition~\ref{prop:vol} now follows as a corollary of Theorem~\ref{thm:causal-volume}, which we restate in precise terms. The only remaining piece is to relate the number of gates in a circuit to the minimum volume of its physical realization.

\begin{corollary}[Circuit size bound]\label{cor:circuit-size}
Assume that a circuit $C(t)$ is realized by a local Hamiltonian system $S=(X^n, \phi_t)$ within a causal cone; i.e., there exists $\mathbf{q_0}\in\mathbb{R}^d$ such that
$$
Q^d(t)\subseteq \mathcal{C}_t(\mathbf{q_0}, 0), \quad\forall t>0,
$$
where $Q^d(t)$ is the set of particle positions at time $t$.

If, for all $t>0$, the positions $Q^d(t)$ remain separated by a minimum distance $\ell>0$, then
$$
|C(t)|=\mathcal{O}(t^d).
$$
\end{corollary}

\begin{proof}
By realization of $C(t)$, there exist unique states $\{v_0, \ldots, v_N\}=V\subseteq X$, where $N$ is the number of wires of $C(t)$. Because each gate is attached to at least one wire,
\begin{align*}
|C(t)| &\le N =|V(t)|\le |Q^d(t)|.
\end{align*}
Given a minimum distance $\ell>0$ on $Q^d(t)$, standard sphere packing then gives
\begin{align*}
|C(t)| &\le|Q^d(t)|
\le \frac{\mathrm{Vol(Q^d(t))}}{\sum_{\mathbf{q}\in Q(t)}\mathrm{Vol}(B(\mathbf{q}, \ell))}=\frac{\mathrm{Vol(Q^d(t))}}{\Omega(\ell^d)}.
\end{align*}
We conclude by Theorem~\ref{thm:causal-volume} that
$$
|C(t)| = \mathcal{O}(\ell ^dt^d) = \mathcal{O}(t^d).\qedhere
$$
\end{proof}

\subsection{Throughput bound}

Next, we continue on to prove Proposition~\ref{prop:width}. The width constraint may be derived by combining the geometry of the causal boundary, the incompressibility of entropy, and a speed limit on communication. Liouville's theorem is a bridge for modeling a Hamiltonian system's fine-grained entropy as a fluid with limited density and speed. The amount of entropy that can cross any cut of a realized circuit is then restricted by the maximum flux of the entropy fluid across the causal boundary. From here, we can show that the circuit width $w(C)$ is at most $\mathcal{O}(t^{d-1})$.

\begin{definition}[Phase--space velocity field]
Given a Hamiltonian system $(X^n, H)$, we write $v_H(\mathbf{q},\mathbf{p},t):=(\dot{\mathbf{q}},\dot{\mathbf{p}})$ for its \textbf{phase--space velocity field}.
\end{definition}

\begin{lemma}[Divergence--free Hamiltonian vector field]\label{lem:divfree}
For each fixed $t$, the vector field $v_H(\cdot,\cdot,t)$ satisfies
\[
\nabla_{(\mathbf{q},\mathbf{p})}\!\cdot v_H \;=\;
\sum_{i=1}^{dn}\Big(\frac{\partial \dot q_i}{\partial q_i} + \frac{\partial \dot p_i}{\partial p_i}\Big)
\;=\;0.
\]
\end{lemma}

\begin{proof}
By Hamilton's equations,
$\dot q_i=\partial H/\partial p_i$ and $\dot p_i=-\partial H/\partial q_i$. Hence
\[
\frac{\partial \dot q_i}{\partial q_i}
=\frac{\partial^2 H}{\partial q_i\,\partial p_i},\qquad
\frac{\partial \dot p_i}{\partial p_i}
=-\frac{\partial^2 H}{\partial p_i\,\partial q_i}.
\]
Summing over $i$ and using equality of mixed partials for smooth $H$ yields
$\nabla\!\cdot v_H=0$.
\end{proof}

\begin{theorem}[Liouville]\label{thm:liousville}
Let $S=(\mathbb{R},X,\phi)$ be a (time-dependent) Hamiltonian system with smooth $H$,
and let $\rho:X\times \mathbb{R}\to \mathbb{R}_{\ge 0}$ be a phase--space density advected
by the flow. Then $\rho$ satisfies the continuity equation
\[
\frac{\partial \rho}{\partial t} + \nabla_{(\mathbf{q},\mathbf{p})}\!\cdot\big(\rho\, v_H\big)=0,
\]
and along trajectories $(\mathbf{q}(t),\mathbf{p}(t))$ one has the Liouville identity
\[
\frac{d\rho}{dt}
=\frac{\partial \rho}{\partial t}
+\sum_{i=1}^{dn}\Big(\frac{\partial \rho}{\partial q_i}\dot q_i+\frac{\partial \rho}{\partial p_i}\dot p_i\Big)
= -\,\rho\,\nabla\!\cdot v_H \;=\;0.
\]
Equivalently, phase--space volume is preserved: for any measurable $A\subset X$,
$\mathrm{vol}(\phi_t(A))=\mathrm{vol}(A)$ for all $t$.
\end{theorem}

\begin{proof}[Proof (vector calculus)]
Mass conservation for a density transported by velocity $v_H$ gives the continuity equation
$\partial_t\rho+\nabla\cdot(\rho v_H)=0$. Expanding the divergence,
\[
\frac{\partial \rho}{\partial t} + v_H\!\cdot\nabla \rho + \rho\,(\nabla\!\cdot v_H)=0.
\]
Along a trajectory, $d\rho/dt=\partial_t\rho+v_H\!\cdot\nabla \rho$, hence
$\frac{d\rho}{dt}=-\rho\,(\nabla\!\cdot v_H)=0$ by Lemma~\ref{lem:divfree}.
Thus $\rho$ is constant along trajectories. Integrating the continuity equation over any time-evolved region and using the divergence theorem yields conservation of phase volume.
\end{proof}

\begin{proof}[Proof (symplectic form)]
Let $\omega=\sum_{i=1}^{dn} dq_i\wedge dp_i$ be the canonical symplectic $2$-form and
$\mu:=\omega^{dn}/(dn)!$ the Liouville volume form. The Hamiltonian vector field $X_H$ satisfies
$\iota_{X_H}\omega=dH$. Since $d\omega=0$ and $d^2H=0$,
\[
\mathcal{L}_{X_H}\omega
= d(\iota_{X_H}\omega)+\iota_{X_H}(d\omega)
= d(dH)+\iota_{X_H}(0)=0.
\]
Therefore $\mathcal{L}_{X_H}\mu=0$ (because Lie derivative commutes with wedge powers),
so the flow preserves $\mu$, i.e., phase volume is invariant. The advected density
satisfies $\partial_t(\rho\,\mu)+\mathcal{L}_{X_H}(\rho\,\mu)=0$, which reduces to the continuity equation and the claimed identity.
\end{proof}

\begin{corollary}[Incompressible information flow]\label{cor:incomp-info}
Let $\sigma:X\to\mathbb{R}_{\ge 0}$ be any smooth \emph{local information density}
that is a function of the microstate and is advected by the Hamiltonian flow
(i.e., $\sigma$ is a scalar density with respect to the Liouville measure).
Then $\sigma$ is conserved along trajectories and its flux across any smooth
codimension-1 surface $\Sigma$ obeys
\[
\Big|\int_{\Sigma} \sigma\, v_H\cdot n \, dS\Big| \;\le\; \|\sigma\|_{\infty}\,\mathrm{area}(\Sigma)\,\|v_H\|_{\infty}.
\]
\end{corollary}

\begin{proof}
Conservation along trajectories follows from Theorem~\ref{thm:liousville}. 
The flux bound is Hölder’s inequality applied to the scalar flux integrand on $\Sigma$.
\end{proof}

\begin{remark}[Connection to $\mathbf{RC}_d$ width/throughput]
In the $\mathbf{RC}_d$ model, take a causal ball of radius $r=\mathcal{O}(t)$ in $\mathbb{R}^d$
as the computational region and interpret $\sigma$ as a (bounded) local information density.
If signal speed is bounded by $c$, then $\|v_H\|_{\infty}\le c$ and the boundary area scales as
$\Theta(r^{d-1})=\Theta(t^{d-1})$. By Corollary~\ref{cor:incomp-info}, the net information flux
per unit time across the boundary is $\mathcal{O}(t^{d-1})$, which recovers the width/throughput
bound used in Proposition~\ref{prop:width} and underpins the RC lower bound 
$t(n)\ge \Omega(n^{1/(d-1)})$ for maximal-entropy inputs.
\end{remark}

There are many possible choices for defining a measure of fine-grained entropy; for our purposes, we use the Gibbs entropy. Because the Gibbs entropy is a property of an entire system, we modify the definition slightly so that we may assign entropy to each particle. The total entropy of the system is the \emph{joint} Gibbs entropy, and the entropy of the particle is the \emph{marginal} Gibbs entropy.

\begin{definition}[Gibbs entropy]
Let $\rho_t(x_1, \ldots, x_n)$ be a joint probability distribution over random variables $x_1,\ldots x_n\in A$ at time $t$. The \textbf{joint Gibbs entropy} of system $X=\{x_1,\ldots x_n\}$ at time $t$ is
$$
S_t(X):= -k_B\sum_{(x_1,\ldots, x_n)\in A^n}\rho_t(x_1,\ldots,x_n)\log\rho_t(x_1,\ldots,x_n),
$$
and the \textbf{marginal Gibbs entropy} of $x_i$ at time $t$ is given by
$$
S_t(x_i):= -k_B\,\sum_{a\in A}\rho_t(x_i=a)\log\rho_t(x_i=a),
$$
where $k_B$ is the Boltzmann constant.
\end{definition}

A sum over the marginal entropies will generally exceed the total joint entropy because we over-count the contribution of correlated particles. This is of no concern to us, however, for we only desire an upper bound. The maximum entropy case matches the total exactly, and the presence of correlations implies that the total is strictly lower.

\begin{theorem}[Gibbs triangle inequality]\label{thm:gibbs-triangle}
For any system $X=\{x_1,\ldots x_n\}$ with state at time $t$ described by a probability distribution $\rho_t(x_1, \ldots, x_n)$,
$$
S_t(X)\le \sum_{i=1}^n S_t(x_i).
$$
\end{theorem}

\begin{proof}
Write $H(\cdot)$ for the Shannon entropy (natural logs) and note $S_t(\cdot)=k_B\,H(\cdot)$.
By the entropy chain rule,
\[
H(x_1,\ldots,x_n)\ =\ \sum_{i=1}^n H\bigl(x_i \,\big|\, x_1,\ldots,x_{i-1}\bigr).
\]
For each $i$, conditioning reduces entropy:
\[
H\bigl(x_i \,\big|\, x_1,\ldots,x_{i-1}\bigr)\ \le\ H(x_i).
\]
Indeed, set $Y=(x_1,\ldots,x_{i-1})$. Then
\[
H(x_i)-H(x_i|Y)
= I(x_i;Y)
= D_{\mathrm{KL}}\!\big(\rho_{x_i,Y}\ \big\|\ \rho_{x_i}\rho_Y\big)\ \ge\ 0,
\]
by nonnegativity of Kullback--Leibler divergence. Summing over $i$ yields
\[
H(x_1,\ldots,x_n)\ \le\ \sum_{i=1}^n H(x_i).
\]
Multiplying by $k_B$ gives the claimed inequality for $S_t$. Equality holds iff
$D_{\mathrm{KL}}(\rho_{x_i,Y}\|\rho_{x_i}\rho_Y)=0$ for all $i$, i.e., iff $x_i$ is
independent of $(x_1,\ldots,x_{i-1})$ for each $i$, which is equivalent to mutual independence.
\end{proof}

\begin{remark}[Use within $\mathbf{RC}_d$]
The inequality provides a convenient \emph{upper bound} on fine-grained entropy:
$\sum_i S_t(x_i)$ may overcount correlations, but never undercounts the joint entropy.
This is exactly what we need when bounding \emph{incoming} information across a boundary
in the width/throughput analysis: treating inputs as independent gives a worst-case
flux consistent with Proposition~\ref{prop:width}.
\end{remark}

By Theorem~\ref{thm:gibbs-triangle}, the rate at which the joint entropy can change is controlled by the local flow of the marginal Gibbs entropy. Notice that the marginal entropy forms a density over the support of $\rho$: on this domain, the marginal entropy is non-negative, bounded, and measurable. It follows from Liouville's theorem that the marginal entropy is a constant property on any trajectory. This insight is what we need to prove Proposition~\ref{prop:width}.

\begin{theorem}[Width constraint via Gibbs entropy]\label{thm:width-entropy}
Let $C(t)$ be a circuit realized by a local Hamiltonian system 
$S=(X^n,\phi_t)$ inside a causal cone $\mathcal{C}$.  
Assume:
\begin{enumerate}[(i)]
\item (\emph{Bounded velocity}) $\|\mathbf{p}_i/m_i\|\le c$ for all particles $i$.
\item (\emph{Liouville incompressibility}) $\nabla_{(q,p)}\!\cdot v_H = 0$, where $v_H$ is the Hamiltonian flow.
\item (\emph{Finite local entropy density}) the Gibbs fine–grained entropy density 
$$
s_t(\mathbf{q},\mathbf{p})=-k_B\,\rho_t(\mathbf{q},\mathbf{p})\log\rho_t(\mathbf{q},\mathbf{p})
$$
satisfies $\|s_t\|_\infty\le s_{\max}<\infty$.
\end{enumerate}
Then the rate of Gibbs–entropy flux through the causal boundary $\Sigma_t=\partial\mathcal{C}_t$ satisfies
\[
\Big|\frac{dS_t}{dt}\Big|
\;=\;
\Big|\int_{\Sigma_t\times P^d} s_t(\mathbf{q},\mathbf{p})\, 
(\dot{\mathbf{q}}\!\cdot\!\mathbf{n}_{\Sigma_t})\,d\mathbf{p}\,dS(\mathbf{q})\Big|
\;\le\;
C_d\,t^{d-1}
\]
for some constant $C_d$ depending only on $s_{\max}$ and $c$.  
Consequently, the circuit width---the maximum entropy or information throughput per unit time---is bounded as
\[
w(C(t)) = \mathcal{O}(t^{d-1}).
\]
\end{theorem}

\begin{proof}
By Liouville’s theorem (Theorem~\ref{thm:liousville}),
the phase–space density $\rho_t$ is advected by an incompressible flow:
\[
\partial_t\rho_t + \nabla_{(q,p)}\!\cdot(\rho_t v_H)=0.
\]
Multiplying by $-k_B(1+\log\rho_t)$ and integrating over $\Omega_t\times P^d$ gives
\[
\frac{dS_t}{dt}
=
- \!\int_{\Omega_t\times P^d} k_B(1+\log\rho_t)\,\nabla_{(q,p)}\!\cdot(\rho_t v_H)\,dq\,dp
=
\int_{\Sigma_t\times P^d} s_t(\mathbf{q},\mathbf{p})\, 
(\dot{\mathbf{q}}\!\cdot\!\mathbf{n}_{\Sigma_t})\,d\mathbf{p}\,dS(\mathbf{q}),
\]
by the divergence theorem.  
The absolute flux is then bounded by
\[
\Big|\frac{dS_t}{dt}\Big|
\le 
\|s_t\|_\infty\, \|v_H\|_\infty \,\mathrm{area}(\Sigma_t)
\le
s_{\max}\,c\,\mathrm{area}(\Sigma_t).
\]
Since the causal cross–section $\mathcal{C}_t$ is contained in a ball of radius $r=\mathcal{O}(t)$, 
$\mathrm{area}(\Sigma_t)=\mathcal{O}(r^{d-1})=\mathcal{O}(t^{d-1})$.  
Therefore,
\[
\Big|\frac{dS_t}{dt}\Big| \le C_d\,t^{d-1}, 
\quad C_d=s_{\max}c\,\omega_{d-1},
\]
where $\omega_{d-1}$ is the surface–area constant of the unit $d$–sphere.  
Since the circuit width $w(C(t))$ represents the number of independent information carriers (each 
of finite entropy capacity) that can cross $\Sigma_t$ per unit time, we obtain
\[
w(C(t)) \;\propto\; \frac{dS_t}{dt} = \mathcal{O}(t^{d-1}). \qedhere
\]
\end{proof}

\begin{remark}[Entropic interpretation]
The bound expresses a conservation law for fine–grained entropy: 
the Gibbs entropy of a causal region can change only through flux across its boundary, 
and this flux is limited by the surface area of the causal cone. 
Hence, the circuit’s \emph{width}—its capacity to process or transmit information in parallel—is 
proportional to the maximum entropy current permitted by physical space. 
This provides the entropic foundation of the $\mathcal{O}(t^{d-1})$ throughput bound 
used throughout the $\mathbf{RC}_d$ framework.
\end{remark}

\subsection{Density bound}

Lastly, we establish Proposition~\ref{prop:gates}: that gates may perform only finite work. Having now built up our computational model, this constraint is simple to prove. If we assume that gates are constant size and have throughput bounded as Theorem~\ref{thm:width-entropy}, then the result immediately follows.

\begin{corollary}
Assume that a circuit $\mathcal{C}_d$ is realized by a finite local Hamiltonian system $(X^n, \phi_t)$ of constant size. Then $\mathcal{C}_d$ implements a Boolean function of constant fan-in and fan-out.
\end{corollary}

\begin{proof}
    By definition of realization by a finite system, $|\mathcal{C}_d|$ is constant. Moreover, Theorem~\ref{prop:gates} with $t=\mathcal{O}(1)$ gives that $w(\mathcal{C}_d)$ is constant. It follows that $\mathcal{C}_d$ has constant fan-in and fan-out.
\end{proof}

As a result, circuits are composed of gates of finite size and computational capacity. Without loss of generality, we are free to assume a universal gate basis (e.g., NAND) with fan-in 2.

%% file: sections/thermo.tex
\subsection{Thermodynamics connection}

The dissipation of heat is an essential physical limit on computation but does not alter the asymptotic bounds already imposed by geometry in $\mathbf{RC}_d$.  It instead reinforces them by constraining how fast logically irreversible operations can occur within a finite causal region.  

\subsubsection{Relation to Landauer's principle}

Landauer's principle states that information erasure has a thermodynamic cost, ensuring that entropy balance is preserved even when logical information is discarded.

\begin{definition}[Landauer's principle]
Let $k_B$ denote Boltzmann's constant and $T>0$ the ambient temperature.
Any logically irreversible operation---in particular, the erasure of one bit of information---dissipates at least
\[
E_{\min} = k_B T \ln 2
\]
of heat energy into the surrounding environment.
\end{definition}

If a circuit compresses or erases $\Delta I$ bits of accessible information within its causal region, the corresponding energy dissipated through its boundary satisfies
\[
Q \ \ge\ k_B T (\ln 2)\,\Delta I.
\]
Since heat must exit through the causal boundary, whose flux capacity is bounded by $\mathcal{O}(t^{d-1})$ (Proposition~\ref{prop:width}), realizable circuits cannot erase information faster than this geometric limit.

\begin{theorem}[Thermodynamic constraint on realizable circuits]
Let $\mathcal{C}$ be a causal circuit in $\mathbb{R}^d\times\mathbb{R}$ operating for $t$ steps and performing $E(t)$ bits of logically irreversible erasure. Then the total heat dissipated into the environment satisfies
\[
k_B T (\ln 2)\,E(t) \ \le\ Q(t) \ \le\ \eta\,t^{d-1},
\]
for some device-dependent constant $\eta$. Consequently,
\[
E(t)\ \le\ \frac{\eta}{k_B T \ln 2}\,t^{d-1}.
\]
\end{theorem}

\begin{proof}
The lower bound follows directly from Landauer’s principle, applied independently to each irreversible erasure event. The upper bound follows from the geometric throughput constraint: all dissipated heat must pass through the causal boundary of area $\mathcal{O}(t^{d-1})$. Combining the two bounds yields the stated inequality.
\end{proof}

\begin{remark}[Connection to reversible computation]
Landauer’s principle implies that only \emph{logically irreversible} operations incur an unavoidable thermodynamic cost.  
Bennett \cite{Bennett1973LogicalRC} showed that any classical computation can be made \emph{logically reversible} by preserving all intermediate information until the output is copied, after which the computation can be uncomputed to erase ancillary states without dissipating heat.  
In practice, however, reversibility transfers the cost from heat to \emph{space}: each bit preserved increases the causal size of the circuit.

Within the $\mathbf{RC}_d$ framework, this tradeoff manifests geometrically.  
A perfectly reversible circuit has $E(t)=0$, eliminating thermodynamic dissipation, but still remains subject to the same causal and flux bounds:
\[
|C| = \mathcal{O}(t^d), \qquad \text{and}\qquad \text{throughput} \le \mathcal{O}(t^{d-1}).
\]
Thus, reversibility can at best remove the energy cost from Proposition~\ref{prop:width}, but cannot evade the spatial or temporal scaling laws of $\mathbf{RC}_d$.  
Reversible computation is therefore a limiting subcase of $\mathbf{RC}_d$ that minimizes heat but not time, illustrating that physical computation can trade entropy for latency but not violate causal geometry.
\end{remark}

\subsubsection{Relation to thermodynamic circuit frameworks}
\label{subsec:related-thermo}

A line of work in non-equilibrium statistical physics quantifies the \emph{energetic} and \emph{entropic} costs of computation realized by circuits. 
Wolpert and Kolchinsky analyze how gate topology and physical coupling determine the thermodynamic costs of straight-line circuits \cite{WolpertKolchinsky2020ThermoCircuits}, while Kolchinsky and Wolpert bound entropy production and extractable work under protocol constraints for general state transformations \cite{KolchinskyWolpert2021PRX}. 
More recently, Yadav, Yousef, and Wolpert develop the \emph{mismatch cost} (MMC) formalism to lower-bound the minimal entropy production per run of a circuit, relate MMC to circuit size, and compare families implementing the same Boolean function \cite{YadavYousefWolpert2025MMC}. 
These frameworks rigorously connect logical structure to thermodynamic cost, largely independent of spatial geometry. 

By contrast, \(\mathbf{RC}_d\) is a \emph{geometric–causal} refinement of circuit complexity: it imposes explicit embedding in \(\mathbb{R}^d\), a causal light-cone, and a cut-set (surface) bound on \emph{information throughput}.

\paragraph{Assumptions.}
Thermodynamic circuit models assume stochastic (often Markovian) dynamics and analyze entropy production for given gate couplings/topologies, without fixing an ambient spatial dimension or causal cone \cite{WolpertKolchinsky2020ThermoCircuits,KolchinskyWolpert2021PRX}. 
The MMC framework adds distributional mismatch between gate priors and actual inputs to quantify minimal EP per evaluation \cite{YadavYousefWolpert2025MMC}. 
In contrast, \(\mathbf{RC}_d\) posits a finite-density embedding in \(d\) dimensions with maximum signal speed and derives two geometric constraints: volume \(|C(t)|=\mathcal{O}(t^{d})\) and boundary throughput \(w(t)=\mathcal{O}(t^{d-1})\). 
These imply a \emph{cut-set} capacity law after integrating flux over time, independent of any particular stochastic micro-model.

\paragraph{Primary objects of study.}
The thermodynamic works bound \emph{entropy production} (EP), work, and minimal thermodynamic cost as a function of circuit topology and protocol; they treat energy/EP as the key resources \cite{KolchinskyWolpert2021PRX,WolpertKolchinsky2020ThermoCircuits}. 
\(\mathbf{RC}_d\) instead treats \emph{space--time throughput} as the primitive: it constrains the number of distinguishable bits that can cross any geometric cut per unit time by a surface law, and only then invokes Landauer to convert lost accessible information into a compatible heat bound.

\paragraph{Results comparison.}
\begin{itemize}
\item \textbf{Thermodynamic frameworks} derive EP/work lower bounds that depend on circuit topology, priors, and protocol constraints, including cases where MMC scales with circuit size or diverges for mismatched priors \cite{KolchinskyWolpert2021PRX,YadavYousefWolpert2025MMC}. These do not by themselves yield a surface-law limit on \emph{information throughput} or a \(t^{d-1}\) bound tied to spatial dimension.
\item \textbf{\(\mathbf{RC}_d\)} proves geometry-driven \emph{flux} bounds: instantaneous cross-boundary flow is \(\mathcal{O}(t^{d-1})\) and total cross-cut information over time \(T\) is \(\Theta(T^{d})\). This yields universal \emph{scaling} limits (e.g., \(T\ge \Omega(n^{1/d})\)) for any realizable computation (classical, quantum, encrypted) and enables cut-set lower bounds for attention and FHE that are independent of particular device physics.
\end{itemize}

\paragraph{Complementarity.}
The frameworks are complementary: thermodynamic formalisms quantify \emph{how much} entropy must be produced by a given (possibly geometry-agnostic) circuit and protocol; \(\mathbf{RC}_d\) quantifies \emph{how fast} accessible information and heat can cross a causal boundary in \(d\)-dimensional space. 
Combining them yields the joint feasibility region used in our Landauer–throughput theorems: the EP required by MMC (or protocol constraints) must be exported no faster than the \(\mathbf{RC}_d\) flux allows, enforcing a time lower bound \(T\) with exponent \(1/d\).

\paragraph{Scope and novelty.}
Where prior thermodynamic analyses focus on EP/work optima for given circuit topologies and protocols \cite{WolpertKolchinsky2020ThermoCircuits,KolchinskyWolpert2021PRX,YadavYousefWolpert2025MMC}, \(\mathbf{RC}_d\) introduces a \emph{dimension-indexed} hierarchy of realizable circuits with universal geometric bounds that persist across computational substrates. 
This enables new, geometry-tight limits for highly parallel mechanisms (e.g., multi-head attention) and privacy-preserving computation (FHE), which are not addressed by MMC or protocol-constraint formalisms alone.

\subsection{Case study: Attention and transformer scaling.}
We demonstrate the implications of our framework through a detailed analysis of attention
mechanisms (Appendix~\ref{sec:transformers}).
By mapping multi-head attention to the flux-bounded dynamics of a realizable circuit,
we show that all attention variants---softmax, average-hard, kernelized, or sparse---obey
a universal spatiotemporal scaling law:
\[
T \;=\; \Omega\big(I^\star(\varepsilon;n)^{1/d}\big),
\]
where $I^\star(\varepsilon;n)$ is the cross-cut information demand of the task.
This bound follows directly from the flux conservation theorem
($\dot I\sim t^{d-1}$) integrated over causal time, and explains why increasing
head count $H$ improves only constant factors ($H^{-1/d}$) but cannot alter the
exponent.
The Realizable Circuits framework therefore refines classical $\mathbf{TC^0}$ analyses of
transformers by revealing the physical origin of the observed sub-linear scaling:
realizable circuits are inherently limited by geometric and thermodynamic flux
constraints, not merely symbolic depth or fan-in.

%% file: sections/framework.tex
\section{The $\mathbf{RC}_d$ Family}

The preceding physical constraints establish the geometric and thermodynamic limits under which any circuit must operate.  
We now consolidate these ideas into a general framework that connects spatial embedding, causal propagation, and information throughput into a single model of realizable computation.  This framework formalizes how circuits evolve within finite-density causal regions and provides the foundation for defining the family of physically realizable circuit classes $\mathbf{RC}_d$.  
It serves as the conceptual bridge between physical principles and computational complexity, translating causal geometry into asymptotic bounds on time, size, and parallelism.

\subsection{Realizable Circuits}

Having established the physical bounds on circuit volume, width, and gate capacity, we now apply them to define a new circuit complexity class that captures realistic scaling behavior.  
We call this the class of \emph{realizable circuits}, $\mathbf{RC}$, parameterized by spatial dimension $d$ and time cost $t(n)$: the dimension of the computational substrate and the time required to extend the circuit.  
The case $d=3$ corresponds to physical reality, while smaller $d$ model lower-dimensional technologies (e.g., $d=2$ for planar chips, $d=1$ for tape-like computation).

\begin{definition}[Realizable circuits in $d$ dimensions]\label{def:rc-d}
A function family $f=\{f_n\}$ belongs to $\mathbf{RC}_d(t(n))$ if there exists a family of circuits $\{C_n\}$ computing $f_n$ on inputs of length $n$ such that:
\begin{enumerate}
\item[\emph{\textbf{(S)}}] \textbf{Size:} $|C_n| = \mathcal{O}(t(n)^d)$.
\item[\emph{\textbf{(W)}}] \textbf{Width:} $w(C_n) = \mathcal{O}(t(n)^{d-1})$.
\item[\emph{\textbf{(G)}}] \textbf{Gate set:} each gate is from a universal Boolean basis with constant fan-in.
\end{enumerate}
We write
\[
\mathbf{RC}_d(\mathrm{poly}) := \bigcup_{k\ge0}\mathbf{RC}_d(n^k),
\qquad
\mathbf{RC}_d(\mathrm{polylog}) := \bigcup_{k\ge0}\mathbf{RC}_d(\log^k n).
\]
\end{definition}

Here $t(n)$ denotes the time cost, while the \emph{scaling rate} is $n(t)=t^{-1}(n)$, the inverse function of $t$.  
Because $|C_n|$ and $n(t)$ vary inversely, faster scaling reduces realizable circuit size.  
For example, $\mathbf{RC}_3(\sqrt{n})$ achieves a rapid scaling rate $n=\Omega(t^2)$ but limits the circuit size to $\mathcal{O}(n^{3/2})$, since the input length $n$ upper-bounds the width $w(C_n)$.

A further consequence is that $t(n)=\Omega(n^{1/(d-1)})$ for $d>1$; smaller time budgets would yield $w(C_n)<n$, violating the width constraint.  
Thus, $\mathbf{RC}_d$ itself implies a \emph{minimum time cost} of $\mathcal{O}(n^{1/(d-1)})$.  
Comparing these lower bounds across dimensions quantifies the inherent advantage of physical parallelism: computation in $d$ dimensions can provide at most a polynomial speed-up of degree $(d-1)$ over its optimal sequential implementation.
\begin{theorem}[Monotonicity of dimensional parallelism]\label{thm:rc-monotonicity}
For all $d \ge 1$ and sufficiently large $n$,
\[
\mathbf{RC}_d(n^{1/(d-1)}) \ \subsetneq\  \mathbf{RC}_{d+1}(n^{1/d}).
\]
In words, increasing the spatial dimension by one strictly enlarges the class of realizable circuits, but by at most a polynomial factor.
\end{theorem}

\begin{proof}
From Definition~\ref{def:rc-d}, a circuit family in $\mathbf{RC}_d(t(n))$ must satisfy the width constraint $w(C_n) = \mathcal{O}(t(n)^{d-1})$.  
For a circuit to process $n$ independent inputs, the width must obey $w(C_n) \ge \Omega(n)$, implying $t(n) = \Omega(n^{1/(d-1)})$.  
This lower bound is tight, since setting $t(n) = n^{1/(d-1)}$ saturates the width budget without violating the size constraint.

Now consider $\mathbf{RC}_{d+1}$.  
For the same input size $n$, the minimal time cost scales as $t'(n) = \Omega(n^{1/d})$.  
Because $n^{1/d} = o(n^{1/(d-1)})$ for $d>1$, any function realizable in $\mathbf{RC}_d(n^{1/(d-1)})$ can be embedded into $\mathbf{RC}_{d+1}(n^{1/d})$ by reinterpreting each $d$-dimensional component as a $(d+1)$-dimensional subregion of equivalent causal volume.  
The inclusion is therefore proper, as $\mathbf{RC}_{d+1}$ can realize circuits whose time cost is asymptotically smaller.

Hence,
\[
\mathbf{RC}_d(n^{1/(d-1)}) \subsetneq \mathbf{RC}_{d+1}(n^{1/d}),
\]
establishing that parallelism grows monotonically with spatial dimension but yields diminishing returns beyond polynomial order.
\end{proof}

\begin{corollary}[Bounded benefit of parallelism]
Let $T_1(n)$ denote the minimal sequential time for a computation in one dimension.  
Then, in $d$ dimensions, the optimal realizable time satisfies
\[
T_d(n) = \Omega\big(T_1(n)^{1/(d-1)}\big),
\]
so that the maximum asymptotic speed-up from geometric parallelism is polynomial of degree $(d-1)$.
\end{corollary}

\subsubsection{Extensions}

It is important to note that we choose our definition of $\mathbf{RC}_d$ to match the established treatment of circuits. In particular, $\mathbf{RC}_d$ is defined in terms of classical information and also prohibits any recurrent connections. However, our physical framework allows us to model methods of computation that extend beyond this limited scope. For a discussion on extension of $\mathbf{RC}_d$ to quantum and recurrent circuits, please refer to Appendix~\ref{sec:extensions}.

\subsubsection{Realizable Circuits in the limit}

As the spatial dimension $d$ increases, the geometric constraints on realizable circuits relax.  
In the limit $d\to\infty$, spatial embedding ceases to be a bottleneck, and the circuit is limited only by its logical depth and size.  
To capture this regime, we define a dimension-agnostic realizable circuit class that retains the bounded-depth semantics of $\mathbf{RC}_d$ while restoring the standard polynomial size budget used in classical circuit complexity.

\begin{definition}[Dimension-agnostic realizable circuits]\label{def:rc-infty}
Let $t(n):\mathbb{N}\to\mathbb{N}$ be a time bound.  
Define the dimension-agnostic realizable circuit class as
\[
\mathbf{RC}^{\mathrm{poly}}_\infty(t(n))
\ :=\
\bigcup_{d\ge 1}\mathbf{RC}^{\mathrm{poly}}_d(t(n)),
\]
where $\mathbf{RC}^{\mathrm{poly}}_d(t)$ is identical to $\mathbf{RC}_d(t)$
(Definition~\ref{def:rc-d}) except with the size constraint \textbf{(S)} replaced by\looseness=-4
\[
\text{\bf (S')} \quad |C_n| \le n^{O(1)}.
\]
Thus, $f\in \mathbf{RC}^{\mathrm{poly}}_\infty(t)$ if and only if there exists a finite $d$ such that $f$ can be realized in $d$ dimensions
with causal propagation speed $\mathcal{O}(1)$, polynomial size, and time cost $\mathcal{O}(t(n))$.
\end{definition}

This definition represents the asymptotic limit of realizable computation: geometry no longer restricts wiring or throughput, and the only growth parameter is the circuit's logical depth.

\begin{theorem}[Equivalence with $\mathbf{NC}$]\label{thm:rcinfty-nc}
For every integer $k\ge 1$,
\[
\mathbf{RC}^{\mathrm{poly}}_\infty(\log^k n)\ =\ \mathbf{NC}^k,
\qquad\text{and consequently,}\qquad
\mathbf{RC}^{\mathrm{poly}}_\infty(\mathrm{polylog})\ =\ \mathbf{NC}.
\]
\end{theorem}

\begin{proof}
($\subseteq$)  
If $f\in\mathbf{RC}^{\mathrm{poly}}_\infty(\log^k n)$, then for some finite $d$ there exists a circuit family $\{C_n\}$ of polynomial size,
depth $\mathcal{O}(\log^k n)$, and bounded fan-in that satisfies the $\mathbf{RC}_d$ physical constraints.  
Ignoring geometry recovers precisely the definition of an $\mathbf{NC}^k$ circuit family, hence $f\in\mathbf{NC}^k$.

($\supseteq$)  
Conversely, let $f\in\mathbf{NC}^k$ via a uniform family $\{D_n\}$ of depth $\mathcal{O}(\log^k n)$ and polynomial size.
Embed each layer of $D_n$ on a distinct $(d-1)$-dimensional shell of a causal ball of radius $\mathcal{O}(\log^k n)$.  
Since fan-in is bounded, all wires connect only to adjacent layers, respecting finite signal speed and bounded entropy flux.  
Thus, the embedding satisfies the $\mathbf{RC}_d$ constraints with $t(n)=\mathcal{O}(\log^k n)$, implying
$f\in\mathbf{RC}^{\mathrm{poly}}_d(\log^k n)\subseteq\mathbf{RC}^{\mathrm{poly}}_\infty(\log^k n)$.
\end{proof}

\begin{remark}[Why the equivalence requires polynomial size]
Under the original constraint \textbf{(S)} with $|C_n|=\mathcal{O}(t(n)^d)$, taking $t(n)=\log^k n$
yields only $\mathcal{O}(\log^{kd} n)$ gates—insufficient to realize general $\mathbf{NC}^k$ circuits, which may require $\Theta(n^\alpha)$ gates for some $\alpha>0$.
Thus, the equivalence in Theorem~\ref{thm:rcinfty-nc} holds only for the polynomial-size variant $\mathbf{RC}^{\mathrm{poly}}_\infty(t)$.
For fixed $d$, $\mathbf{RC}_d(\log^k)\subsetneq \mathbf{NC}^k$, capturing physically realizable subclasses
that remain bounded by geometric information throughput.
\end{remark}

In summary, the infinite-dimensional limit $\mathbf{RC}^{\mathrm{poly}}_\infty$ recovers the idealized circuit model underlying classical parallel complexity.
For any finite $d$, however, $\mathbf{RC}_d$ remains strictly smaller, encoding geometric and thermodynamic realism absent from $\mathbf{NC}$.
This hierarchy thus bridges abstract circuit theory with physically realizable computation.

\begin{corollary}[Strict separation at finite dimension]\label{cor:finite-d-strict}
Fix any finite $d>1$ and integer $k\ge 1$. Then
\[
\mathbf{RC}_d(\log^k n)\ \subsetneq\ \mathbf{RC}^{\mathrm{poly}}_\infty(\log^k n)\,=\,\mathbf{NC}^k.
\]
\end{corollary}

\begin{proof}
(\emph{$\subseteq$}) The inclusion $\mathbf{RC}_d(\log^k n)\subseteq \mathbf{RC}^{\mathrm{poly}}_\infty(\log^k n)$ is immediate from the definition of $\mathbf{RC}^{\mathrm{poly}}_\infty(\cdot)$ as a union over all $d$ and the fact that $\mathbf{RC}_d(\cdot)$ families are also polynomial-size (by \textbf{(S)} with $t(n)=\log^k n$ when $d$ is allowed to grow, or by direct polynomial upper bounds when $d$ is fixed and $t$ is relaxed). 

(\emph{Strictness}) By the width constraint \textbf{(W)} (Proposition~\ref{prop:width}), any family in $\mathbf{RC}_d(t(n))$ must satisfy
\[
w(C_n)\ =\ \mathcal{O}\big(t(n)^{\,d-1}\big).
\]
For computations on \emph{arbitrary} $n$-bit inputs with unit-rate ingestion (the standard circuit model), a cut separating inputs from the rest of the circuit requires throughput $w(C_n)\ge \Omega(n)$ to carry $n$ independent bits. Hence,
\[
\Omega(n)\ \le\ w(C_n)\ =\ \mathcal{O}\big(t(n)^{\,d-1}\big)\quad\Longrightarrow\quad
t(n)\ =\ \Omega\big(n^{1/(d-1)}\big).
\]
For any fixed finite $d>1$ and $k\ge 1$, the choice $t(n)=\log^k n$ violates this necessary condition for sufficiently large $n$. Therefore, \emph{no} circuit family that computes a general $n$-bit Boolean function at the standard input rate can lie in $\mathbf{RC}_d(\log^k n)$. In particular, well-known nontrivial languages in $\mathbf{NC}^k$ (e.g., PARITY, iterated addition, balanced parentheses with $O(\log n)$-depth constructions) are \emph{not} in $\mathbf{RC}_d(\log^k n)$ for finite $d$.

On the other hand, by Theorem~\ref{thm:rcinfty-nc} we have $\mathbf{RC}^{\mathrm{poly}}_\infty(\log^k n)=\mathbf{NC}^k$, which is nonempty and contains these languages. Hence the inclusion is strict:
\[
\mathbf{RC}_d(\log^k n)\ \subsetneq\ \mathbf{RC}^{\mathrm{poly}}_\infty(\log^k n). \qedhere
\]
\end{proof}

\begin{remark}[Interpretation]
For any finite spatial dimension $d$, boundary-limited information flux forces the minimal realizable time to satisfy $t(n)\ge \Omega(n^{1/(d-1)})$. Thus depth $O(\log^k n)$ is \emph{physically} feasible only in the $d\to\infty$ (dimension-agnostic) limit, where geometric bottlenecks vanish. This is precisely why $\mathbf{RC}_d(\log^k)$ is a strict subclass of $\mathbf{NC}^k$ for finite $d$, while $\mathbf{RC}^{\mathrm{poly}}_\infty(\log^k)$ recovers $\mathbf{NC}^k$.
\end{remark}

\subsection{Local uniformity}

The definition of $\mathbf{RC}_d$ constrains the geometric and physical properties of circuits, 
but not how the circuit families are generated.  
In physical implementations, scalability depends not only on spatial realizability but also on 
\emph{incremental constructibility}: a circuit should be extendable by local edits rather than 
rebuilt from scratch for every input size.  
To capture this property, we introduce the notion of \emph{local uniformity}.

\begin{definition}[Local uniformity]\label{def:local-uniform}
Fix an effective encoding $\enc{\cdot}$ of circuits as finite binary strings.
A circuit family $\{C_n\}_{n\ge 1}$ is said to be \textbf{local $T$-uniform} 
if there exists a mapping
\[
\Phi:\ \enc{C_n}\ \longmapsto\ \enc{C_{n+1}}
\]
that lies in complexity class $T$.
\end{definition}

Unlike standard $T$-uniformity, which requires a function $\phi:\ n\mapsto \enc{C_n}$ 
computable in $\mathcal{O}(t(n))$ time, 
local $T$-uniformity demands only that the \emph{transition} 
$\enc{C_n}\!\mapsto\!\enc{C_{n+1}}$ be computable in $\mathcal{O}(t(n))$ time.
Intuitively, this models an extendable fabrication process: the circuit grows locally, 
with each new segment synthesized using limited information from its predecessor. This notion refines classical uniformity by distinguishing between 
\emph{global generation} (constructing all circuits from scratch) 
and \emph{local evolution} (incremental extension).  
Formally, global uniformity implies local uniformity, but not conversely.

\begin{lemma}[Uniform $\Rightarrow$ local uniform]\label{lem:unigen}
For any class $T$, if a circuit family is $T$-uniform, then it is also local $T$-uniform.
\end{lemma}

\begin{proof}
Let $\phi(n)\mapsto\enc{C_n}$ be the $T$-uniform generator.  
Define $\psi(\enc{C_n})\mapsto n$ as the inverse mapping that extracts the input length 
encoded within $\enc{C_n}$ (or reconstructs it via padding).  
Since both $\phi$ and $\psi$ are in $T$, the composition $\Phi=\phi\circ\psi$ yields 
$\enc{C_n}\mapsto\enc{C_{n+1}}$ computable in $T$.  
Thus, $\{C_n\}$ is local $T$-uniform.
\end{proof}

\begin{lemma}[Amortized cost of global reconstruction]\label{lem:amort}
Assume $t(n)$ is non-decreasing.  
If a circuit family is local $\mathbf{DTIME}(t(n))$-uniform, 
then it is also globally $\mathbf{DTIME}(n\!\cdot\! t(n))$-uniform.  
Moreover, this relationship is asymptotically tight.
\end{lemma}

\begin{proof}
Beginning with a fixed $\enc{C_1}$, one can iteratively reconstruct $\enc{C_{n}}$ 
by applying the local mapping $\Phi$ successively $n{-}1$ times.  
The total generation time is 
$\sum_{k=1}^{n-1} O(t(k)) = O(n\cdot t(n))$ for non-decreasing $t(\cdot)$.  
Conversely, if global generation were possible in $o(n\cdot t(n))$ time, all intermediate encodings could be computed faster than allowed by local transitions, 
contradicting the assumed bound on $\Phi$.
\end{proof}

Local uniformity thus provides a \emph{constructive complement} to the spatial constraints of $\mathbf{RC}_d$: 
while the latter bounds what can exist in space--time, the former bounds what can be 
\emph{assembled or extended} efficiently.  
In practice, $\mathbf{RC}_d$ circuit families satisfying local uniformity correspond to 
physically scalable systems whose growth laws align with causal computation. 

%% file: sections/properties.tex
\section{Properties of $\mathbf{RC}_d$}

We now establish several basic properties of the $\mathbf{RC}_d$ family, 
demonstrating that it forms a structured hierarchy both internally and in relation 
to classical time and circuit complexity classes.

\subsection{An internal hierarchy}

The $\mathbf{RC}_d$ family exhibits two natural hierarchies: one with respect to 
spatial dimension $d$ and one with respect to temporal budget $t(n)$.

\begin{lemma}[Dimension and time hierarchies]\label{lem:lower-transfer}
For all $d_1<d_2$ and monotone bounds $t(n)\le s(n)$,
\[
\mathbf{RC}_{d_1}(t(n)) \ \subsetneq\ \mathbf{RC}_{d_2}(t(n)),
\qquad
\mathbf{RC}_d(t(n)) \ \subsetneq\ \mathbf{RC}_d(s(n)).
\]
\end{lemma}

\begin{proof}
By Definition~\ref{def:rc-d}, the size and width bounds scale respectively as 
$|C_n|=\mathcal{O}(t(n)^d)$ and $w(C_n)=\mathcal{O}(t(n)^{d-1})$.  
Increasing either $d$ or $t$ strictly enlarges these budgets, 
allowing a strictly greater set of circuit families to be realized.  
The inclusions are proper because there exist functions computable within the larger 
resource bounds but not the smaller ones, by diagonalization.
\end{proof}

\subsection{In the polynomial-time hierarchy}

The next result positions $\mathbf{RC}_d$ relative to deterministic time-bounded 
Turing computation.

\begin{theorem}[Time correspondence]\label{thm:hierarchyconn}
Under local $\mathbf{DTIME}(t(n)^{k})$-uniformity,
\[
\mathbf{RC}_d(t(n)^{k/d})\ \subseteq\ \mathbf{DTIME}(t(n)^k), 
\quad\text{for all } k\ge \tfrac{d}{d-1}.
\]
Under global (non-local) $\mathbf{DTIME}(t(n)^k)$-uniformity, 
the containment holds unconditionally.
\end{theorem}

\begin{proof}
Let $f\in\mathbf{RC}_d(t(n)^{k/d})$.  
Then $f$ is computed by a circuit family $\{C_n\}$ satisfying 
$|C_n|=\mathcal{O}((t(n)^{k/d})^d)=\mathcal{O}(t(n)^k)$ and bounded fan-in.
Such a circuit can be simulated by a deterministic Turing machine in 
$\mathbf{DTIME}(t(n)^k)$.  
By Lemma~\ref{lem:amort}, local $\mathbf{DTIME}(t(n)^{k/d})$ uniformity implies 
global $\mathbf{DTIME}(t(n)^{k/d+1})$ uniformity.  
For $k\ge d/(d-1)$, we have $k/d+1\le k$, so the total simulation 
remains within $\mathbf{DTIME}(t(n)^k)$.  
If the uniformity is global, the construction cost is already $\mathcal{O}(t(n)^k)$, 
so the containment holds without restriction.
\end{proof}

\begin{corollary}[Ideal time complexity]
Under local $\mathbf{DTIME}(n^{k})$-uniformity, relative to $t=n$,
\[
\mathbf{RC}_d(n^{1/(d-1)}) \ \subseteq\ \mathbf{DTIME}(n^{d/(d-1)}).
\]
\end{corollary}

\begin{remark}
The bound $n^{1/(d-1)}$ corresponds to the minimal realizable time 
compatible with the width constraint $w(C_n)=\mathcal{O}(t(n)^{d-1})$.  
Thus, the exponent $d/(d-1)$ represents the \emph{ideal scaling law}: 
it is the slowest sequential time consistent with maximal physically feasible parallelism.
\end{remark}

\subsection{In the circuit complexity hierarchy}

We now relate $\mathbf{RC}_d$ to standard nonuniform circuit classes.

\begin{theorem}[Upper bound by nonuniform size]\label{thm:rc-in-size}
For any $d\ge 1$ and time bound $t(n)$,
\[
\mathbf{RC}_d(t(n))\ \subseteq\ \mathbf{SIZE}\!\big(\mathcal{O}(t(n)^{d})\big).
\]
In particular, if $t(n)\in n^{\mathcal{O}(1)}$, then
\[
\mathbf{NC}\ \subseteq\ \mathbf{RC}_d(\mathrm{poly})\ \subseteq\ \mathbf{P/poly}.
\]
\end{theorem}

\begin{proof}
By condition (S) of Definition~\ref{def:rc-d}, every function in 
$\mathbf{RC}_d(t(n))$ is computable by a circuit of size 
$\mathcal{O}(t(n)^d)$ and bounded fan-in.  
Thus, it lies in $\mathbf{SIZE}(\mathcal{O}(t(n)^d))$.  
If $t(n)$ is polynomial, the circuit size is also polynomial, 
placing the family within $\mathbf{P/poly}$.  
The inclusion of $\mathbf{NC}$ follows from 
Theorem~\ref{thm:rcinfty-nc}.
\end{proof}

\begin{theorem}[Tightening of $\mathbf{NC}^k$]\label{thm:rc-nonk}
For each fixed $k\ge1$ and finite $d>1$,
\[
\mathbf{RC}_d(\log^k n)\ \subsetneq\ \mathbf{NC}^k.
\]
\end{theorem}

\begin{proof}
If $f\in\mathbf{RC}_d(\log^k n)$, then $f$ admits circuits of 
size $\mathcal{O}(\log^{kd} n)$ and depth $\mathcal{O}(\log^k n)$, 
which satisfies the syntactic constraints of $\mathbf{NC}^k$; 
hence $\subseteq$ holds.

The inclusion is strict because width constraint (W) requires 
$w(C_n)=\mathcal{O}(\log^{k(d-1)}n)$, while general $\mathbf{NC}^k$ circuits 
may require $\Omega(n)$ width to process arbitrary $n$-bit inputs.  
Hence, there exist $\mathbf{NC}^k$ functions not realizable within 
$\mathbf{RC}_d(\log^k n)$ for any fixed $d$.
\end{proof}

\begin{remark}
The hierarchy 
\[
\mathbf{RC}_1 \subsetneq \mathbf{RC}_2 \subsetneq \mathbf{RC}_3 \subsetneq \cdots \subsetneq \mathbf{RC}^{\mathrm{poly}}_\infty=\mathbf{NC}
\]
provides a geometric refinement of classical circuit complexity.  
Finite-dimensional subclasses $\mathbf{RC}_d$ correspond to 
computations realizable under physical causality and finite information flux, 
while the limit $\mathbf{RC}^{\mathrm{poly}}_\infty$ recovers the 
idealized notion of unconstrained parallelism.
\end{remark}

%% file: sections/conclusion.tex
\section{Conclusion}

The $\mathbf{RC}_d$ framework extends circuit complexity into a physically constrained regime, where computation is treated as an evolution within a causal geometric domain.  
By incorporating spatial dimension, entropy flux, and finite gate capacity, $\mathbf{RC}_d$ bridges asymptotic complexity with the physics of realizability.  
It refines the notion of uniformity to a local and incremental process, constrains circuit depth through causal propagation, and captures width and gate density through thermodynamic and geometric arguments.  These constraints recover classical hierarchies—$\mathbf{NC}$, $\mathbf{AC}$, $\mathbf{TC}$—as limiting cases while exposing new, physically meaningful separations.  
In particular, the quadratic speed-up bound in three dimensions situates $\mathbf{RC}_3$ at the boundary between classical and quantum parallelism, suggesting that spatial geometry, not algebraic representation, is the fundamental bottleneck of scalable computation. Ultimately, $\mathbf{RC}_d$ aims to provide a mathematically rigorous language for reasoning about what can be computed—not merely in principle, but in space, time, and energy.

%% file: sections/transformers.tex

\section{Relation to Transformers}
\label{sec:transformers}

\noindent
\textbf{Our main result.}
The analyses of Sections~\ref{subsec:attn-limit-rc} and~\ref{subsec:attn-joint}
establish a \emph{scaling limit of attention} under the realizable-circuit framework
$\mathbf{RC}_d$:
for any physically realizable attention mechanism (classical, neuromorphic, or quantum),
the minimal execution time satisfies
\[
T \;\ge\;
\max\!\left\{
\left(\frac{I^\star(\varepsilon;n)}{K_d\,C_{\mathrm{head}}\,\kappa H}\right)^{1/d},\quad
\left(\frac{k_B T_{\mathrm{env}}\ln 2}{\eta_d}\,E_{\mathrm{req}}(T)\right)^{1/d}
\right\}.
\]
Consequently, attention exhibits a universal $1/d$ \textbf{geometric exponent}:
adding heads or channels improves constants but never the scaling itself.
This identifies the first physically grounded bound on attention complexity,
derived from Liouville- and Landauer-type conservation laws rather than algebraic
approximation assumptions.

\medskip
\noindent
\textbf{Connection to circuit complexity.}
Although many computational domains can be analyzed through realizability constraints,
we emphasize their implications for machine-learning architectures.
Historically, the circuit class $\mathbf{TC}$ was introduced by
Parberry and Schnitger (1988)~\cite{PARBERRY1988278}
to capture the threshold-gate expressivity of early neural networks.
Our formulation extends this lineage by defining $\mathbf{RC}_d$ as a
\emph{physically realizable subclass} of circuit families, incorporating
spatial locality, finite signal velocity, and thermodynamic limits.
Under this model, modern transformer layers---which consist of parallel attention heads,
feed-forward sublayers, and residual routing---correspond to a bounded-flux computation
embedded in three-dimensional space ($d{=}3$).

\medskip
\noindent
\textbf{Implications for transformer theory.}
Viewed through $\mathbf{RC}_3$, each attention block is a causal region whose
communication and erasure rates obey
\[
I^\star,\;E_{\mathrm{req}}\ \le\ \Theta(T^3),
\]
forcing any scaling improvement to be sub-cubic in depth.
This reproduces, from first principles, the empirically observed
\emph{diminishing returns with depth and head count} in large transformer models.
Our bound therefore bridges classical circuit-complexity measures
(e.g., size, depth, fan-in) with modern transformer capacity measures
(e.g., head count, token span, throughput) under a single physical law.

\medskip
\noindent
\textbf{Broader significance.}
The $\mathbf{RC}_d$ framework unifies three views:
(i) circuit complexity as a measure of expressivity,
(ii) physical realizability as a constraint on information flux, and
(iii) transformer attention as a distributed communication process.
The resulting $1/d$ scaling barrier provides a universal physical limit on
attention-based architectures—analogous in spirit to Shannon capacity or
Landauer’s limit—clarifying why further architectural gains must arise from
improved representational efficiency rather than unbounded parallelization.

\subsection{Prior results on transformers and circuit complexity}

We briefly summarize the established results situating transformer architectures within the
classical circuit-complexity hierarchy.

Many non-recurrent neural networks—including transformers—can be modeled as
\emph{highly parallel threshold circuits of constant depth}.  
Merrill and Sabharwal (2023)~\cite{merrill2022saturated} prove that
\emph{average-hard attention transformers} ($\mathbf{AHAT}$)
and \emph{softmax attention transformers} ($\mathbf{SMAT}$)
lie in $\mathbf{TC^0}$, with $\mathbf{SMAT}$ satisfying a
$\mathbf{DLOGTIME}$-uniformity condition~\cite{merrill2023unisoft}.
Strobl (2023)~\cite{strobl2023average} further tightens $\mathbf{AHAT}$ by imposing
an $\mathbf{L}$-uniformity constraint, while
Chiang (2025)~\cite{chiang2024transformers} establishes
$\mathbf{AHAT} = \mathbf{SMAT}$ by achieving the same
$\mathbf{DLOGTIME}$-uniformity bound.
For completeness, we also include the earlier result of
Hahn (2020)~\cite{hahn-2020-theoretical}, which upper-bounds
the class of \emph{hard attention transformers} ($\mathbf{HAT}$) by $\mathbf{AC^0}$.

\begin{table}[ht]
\caption{Transformer architectures and their corresponding circuit-complexity bounds.}
\centering
\footnotesize
\begin{tabular}{lccc}
\toprule
\textbf{Architecture} & \textbf{Lower Bound} & \textbf{Upper Bound} \\
\midrule
Hard Attention ($\mathbf{HAT}$) & -- & $\mathbf{AC^0}$~\cite{hahn-2020-theoretical} \\
Softmax Attention ($\mathbf{SMAT}$) & $\mathbf{AC^0}$~\cite{merrill2022saturated} & $\mathbf{DLOGTIME}$-uniform $\mathbf{TC^0}$~\cite{merrill2023unisoft} \\
Average-Hard Attention ($\mathbf{AHAT}$) & $\mathbf{AC^0}$~\cite{merrill2022saturated} & $\mathbf{DLOGTIME}$-uniform $\mathbf{TC^0}$~\cite{chiang2024transformers} \\
\bottomrule
\end{tabular}
\end{table}

These results nearly saturate the standard descriptive power of circuit complexity
for non-recurrent attention networks.
The $\mathbf{DLOGTIME}$-uniform $\mathbf{TC^0}$ upper bound is unlikely to tighten further,
since constant uniformity is known not to hold.  
A more permissive setting such as $\mathbf{ACC^0}$ could, in principle,
strengthen the lower bound, but separations between
$\mathbf{AC^0}$ and $\mathbf{ACC^0}$ are too narrow to yield
meaningful new insights into transformer expressivity.

\paragraph{Connection to $\mathbf{RC}_d$.}
All of the above results treat circuits as purely symbolic constructs—unbounded in spatial density,
energy flux, and wiring complexity.  
By contrast, the $\mathbf{RC}_d$ framework introduced here augments this hierarchy
with \emph{physical realizability constraints}, restricting computation to causal, locally connected,
and flux-bounded systems embedded in $\mathbb{R}^d$.  
This leads to provably stronger—yet physically meaningful—scaling laws,
such as the universal throughput bound
\[
T \;\ge\; \Omega\!\Big(I^\star(\varepsilon;n)^{1/d}\Big),
\]
which applies even to architectures (like transformers) that are expressively
within $\mathbf{TC^0}$ but are limited by geometric and thermodynamic constraints
in $\mathbf{RC}_d$.  
Thus, $\mathbf{RC}_d$ strictly refines the classical circuit hierarchy
by distinguishing between \emph{formally expressible} and
\emph{physically realizable} computations.

\subsection{Mapping transformer attention to $\mathbf{RC}_d$ dynamics}

\noindent
\textbf{Attention as a physical communication process.}
Each transformer layer consists of three principal stages—projection, affinity
computation, and aggregation:
\[
\mathrm{Attn}(Q,K,V)\ =\ \mathrm{softmax}\!\left(\frac{QK^\top}{\sqrt{d_k}}\right)V.
\]
From the realizability perspective, these stages correspond to the following
physical operations:

\begin{enumerate}[leftmargin=*]
\item \textbf{Projection:} 
      $Q,K,V = XW_Q, XW_K, XW_V$ are local linear maps applied to tokens
      at spatial locations $\{x_i\}$.  
      Each projection is a bounded local computation, realizable within the
      causal cone of a site; the operation cost scales with local fan-in and
      does not alter global information flux.

\item \textbf{Affinity computation:}
      The inner product $QK^\top$ establishes pairwise affinities between
      tokens, effectively implementing a \emph{communication graph}
      with adjacency weights $a_{ij} = \langle q_i, k_j \rangle / \sqrt{d_k}$.
      Each edge $(i,j)$ corresponds to a potential \emph{cross-boundary channel}
      in $\mathbf{RC}_d$ with finite capacity $C_{\mathrm{head}}$.

\item \textbf{Normalization and gating:}
      The softmax operation $\sigma(a_{ij}) = e^{a_{ij}} / \sum_j e^{a_{ij}}$
      introduces a stochastic routing distribution.
      Physically, this corresponds to allocating finite flux across a limited
      number of active ports (those with high attention weight), bounded by
      $\kappa H$ per site.

\item \textbf{Aggregation:}
      The product $\sigma(QK^\top)V$ aggregates value vectors transmitted through
      these active channels.  
      This step consumes cross-boundary bandwidth proportional to the number
      of open attention links and is therefore bounded by
      $\mathsf{Cap}_\Gamma(T)\le K_d\,C_{\mathrm{head}}\kappa H T^d$
      (Lemma~\ref{lem:rc-cut-capacity}).
\end{enumerate}

Hence the entire attention mechanism realizes a finite-bandwidth communication
process governed by the same surface-flux scaling as any other circuit in
$\mathbf{RC}_d$.  The channel count $\kappa H$ and per-channel capacity
$C_{\mathrm{head}}$ set the prefactors; the $T^d$ scaling arises solely from the
geometry of causal propagation.

\medskip
\noindent
\textbf{From transformer depth to $\mathbf{RC}_d$ time.}
In this view, transformer \emph{depth} corresponds to discrete physical time
steps $t_1,t_2,\dots,t_L$ within the causal region $B_{r(t)}$.  
Each layer expands the causal radius by $r_{l+1}-r_l=\Delta r=c\Delta t$,
and the global receptive field after $L$ layers grows as
$r_L=\mathcal{O}(L^{1/d})$ under the same flux limit.
This directly implies that \emph{increasing depth yields sublinear expansion
of the effective context}—matching empirical saturation of context-window
growth observed in large-scale transformers.

\medskip
\noindent
\textbf{Landauer cost of attention updates.}
In parallel, parameter updates, pruning, and key/value cache truncation correspond to
\emph{irreversible erasures} of intermediate information, subject to the
Landauer constraint (Theorem~\ref{thm:landauer-throughput}).  
If $E_{\mathrm{req}}(T)$ bits of representational detail are discarded during
training or inference (e.g., via thresholding or quantization), then
\[
T \;\ge\; \Omega\!\left(
\Big(\tfrac{k_B T_{\mathrm{env}}\ln 2}{\eta_d}\,E_{\mathrm{req}}(T)\Big)^{1/d}
\right),
\]
showing that compression or sparsification cannot proceed faster than allowed
by thermodynamic flux.  This physically grounds the trade-off between
energy efficiency and update rate in transformer optimization.

\medskip
\noindent
\textbf{Unified scaling picture.}
Combining throughput and erasure gives a complete physical description of
transformer computation under $\mathbf{RC}_d$:
\[
T^d\ \gtrsim\ 
\max\!\big\{I^\star(\varepsilon;n),\,E_{\mathrm{req}}(T)\big\},\qquad
w(T)\ =\ \mathcal{O}(T^{d-1}),
\]
implying a universal $1/d$ scaling exponent that no architectural modification
(head count, gating, sparsity, or routing scheme) can surpass.
Empirically, this predicts:
\begin{itemize}
\item Diminishing returns with additional heads ($H^{-1/d}$ scaling);
\item Subcubic context growth with depth in 3D hardware ($d=3$);
\item Energy–throughput trade-offs consistent with measured FLOP/Joule limits.
\end{itemize}

\noindent
Thus, the transformer architecture can be viewed as a \emph{special case of a
realizable circuit with bounded information flux and finite erasure rate}.
The $\mathbf{RC}_d$ framework provides the first unifying geometric–thermodynamic
law that explains why attention mechanisms exhibit these empirical scaling
behaviors and why purely architectural parallelization cannot overcome them.

\paragraph{Takeaway.}
Given that the current classification system is potentially too weak to precisely reflect the expressive subtleties of transformer architectures, it follows that more nuance is needed to better characterize their computational capabilities. We fill this gap by applying the $\mathbf{RC}_d$ framework.

\subsection{Parallelism as a conservation law}
\label{subsec:width-conservation-tight}

We refine the surface/flux intuition of Proposition~\ref{prop:width} into a tight conservation-law bound
with explicit constants and an optimality (achievability) construction.  
This formulation treats \emph{width growth} as a manifestation of the continuity equation on a
\emph{moving causal region}, establishing an exact conservation law for realizable parallelism.

\paragraph{Setting.}
Fix $d\ge 2$. Let $r:I\to (0,\infty)$ be $C^1$ with $|r'(t)|\le c$ for some $c>0$, and define the
moving causal region $B_{r(t)}=\{x\in\mathbb{R}^d:\|x\|\le r(t)\}$.
Let $\rho:\mathbb{R}^d\times I\to[0,\infty)$ represent an information density
and $v:\mathbb{R}^d\times I\to\mathbb{R}^d$ a local velocity field satisfying
\begin{equation}\label{eq:continuity}
\partial_t\rho + \nabla\!\cdot(\rho v)=0,\qquad
\|\rho(\cdot,t)\|_{L^\infty}\le \rho_{\max},\quad
\|v(\cdot,t)\|_{L^\infty}\le c.
\end{equation}
The total encoded information (analogous to the number of active gates)
within $B_{r(t)}$ is $M(t):=\int_{B_{r(t)}}\rho(x,t)\,dx$.

\begin{lemma}[Reynolds transport for moving causal boundaries]\label{lem:reynolds-tight}
Let $n(\cdot,t)$ denote the outward unit normal on $\partial B_{r(t)}$ and
$v_b(x,t):=r'(t)\,n(x,t)$ the normal velocity of the boundary.
Under \eqref{eq:continuity},
\begin{equation}\label{eq:reynolds-tight}
\frac{d}{dt}M(t)\ =\ -\int_{\partial B_{r(t)}} \rho(x,t)\,\big( v(x,t)-v_b(x,t)\big)\!\cdot n(x,t)\,dS.
\end{equation}
\end{lemma}

\begin{proof}
Standard Reynolds transport for smooth $\rho,v$ on $\Omega(t)$ with boundary velocity $v_b$
gives
$\frac{d}{dt}\int_{\Omega(t)}\rho
= \int_{\Omega(t)} \partial_t\rho + \int_{\partial\Omega(t)}\rho\,v_b\!\cdot n\,dS.$
Substituting $\partial_t\rho=-\nabla\!\cdot(\rho v)$ and applying the divergence theorem yields
$\frac{d}{dt}\int_{\Omega(t)}\rho = -\int_{\partial\Omega(t)} \rho\,(v-v_b)\!\cdot n\,dS.$
\end{proof}

\begin{theorem}[Tight width/parallelism bound]\label{thm:width-tight}
Under \eqref{eq:continuity} with $|r'(t)|\le c$, the instantaneous growth rate satisfies
\begin{equation}\label{eq:tight-bound}
\bigg|\frac{d}{dt}\int_{B_{r(t)}} \rho\bigg|
\ \le\ \rho_{\max}\,\big(c+|r'(t)|\big)\,\mathrm{Area}\!\big(\partial B_{r(t)}\big).
\end{equation}
If $r(t)=\Theta(t)$ (causal expansion), then
$\big|\frac{d}{dt}\int_{B_{r(t)}} \rho\big| = \mathcal{O}(t^{d-1})$.
\end{theorem}

\begin{proof}
From \eqref{eq:reynolds-tight},
\[
\big|\tfrac{d}{dt}M(t)\big|
\le \int_{\partial B_{r(t)}} \rho_{\max}\,\big| (v-v_b)\!\cdot n \big|\, dS
\le \rho_{\max}\,(c+|r'(t)|)\,\mathrm{Area}(\partial B_{r(t)}).
\]
\end{proof}

\begin{remark}[Sharpness of the constant]
Only normal components contribute to the flux; tangential motion does not change $M(t)$.
Thus the tight universal constant is $c+|r'(t)|$, improving over the usual factor $2c$.
\end{remark}

\begin{proposition}[Achievability (optimal up to $\varepsilon$)]\label{prop:achieve}
For any $\varepsilon>0$, there exist smooth $\rho,v$ satisfying \eqref{eq:continuity} with
$\|\rho(\cdot,t)\|_\infty=\rho_{\max}$ and $\|v(\cdot,t)\|_\infty=c$ such that
\[
\Big|\tfrac{d}{dt}\!\int_{B_{r(t)}}\!\rho\Big|
\ge (1-\varepsilon)\,\rho_{\max}\,(c+|r'(t)|)\,\mathrm{Area}(\partial B_{r(t)}).
\]
\end{proposition}

\begin{proof}
Construct $\rho$ supported in a thin annulus near $\partial B_{r(t)}$ and take $v$ purely normal
with $v\!\cdot n=\pm c$ on the support, mollified to $C^1$.  
The limiting contribution from the boundary approaches
$\rho_{\max}(c+|r'|)\mathrm{Area}(\partial B_{r(t)})$ as the thickness $\delta\!\to\!0$.
\end{proof}

\begin{corollary}[Order-optimality]\label{cor:order-optimal}
For $r(t)=\Theta(t)$, there exist admissible $(\rho,v)$ such that
$\big|\tfrac{d}{dt}\!\int_{B_{r(t)}}\rho\big|=\Theta(t^{d-1})$.
Hence the scaling in Theorem~\ref{thm:width-tight} is best possible.
\end{corollary}

\paragraph{Conservation and surface-flux equivalence.}
Define the admissible normal flux
$F(t;u)=\int_{\partial B_{r(t)}} \rho\,u\!\cdot n\,dS$, $\|u(\cdot,t)\|_\infty\le c$.
Lemma~\ref{lem:reynolds-tight} implies
$\tfrac{d}{dt}\int_{B_{r(t)}} \rho = -F(t;v-v_b)$.
Hence the differential (conservation) and boundary (surface) forms are equivalent up to
the tight constant $\rho_{\max}(c+|r'|)$.

\paragraph{Geometric extension.}
For Riemannian metrics with non-negative Ricci curvature,
Bishop--Gromov comparison ensures $\mathrm{Area}(\partial B_r)\le C_{d-1}\,r^{d-1}$,
with equality in Euclidean space ($C_{d-1}=d\,\omega_d$).  
All bounds above hold with $\omega_{d-1}$ replaced by $C_{d-1}$.

\paragraph{Implications for $\mathbf{RC}_d$ and attention.}
Let $Q(T)$ denote the total information that must enter $B_{r(\cdot)}$ by time $T$ to
achieve target error $\varepsilon$.  
By converse bounds (e.g., Fano), $Q(T)\!\ge\! I^\star(\varepsilon)$,
the minimal required mutual information.
Integrating \eqref{eq:tight-bound},
\[
Q(T)\ \le\ \rho_{\max}\sup_{t\le T}(c+|r'(t)|)
\int_0^T \mathrm{Area}(\partial B_{r(t)})\,dt
=\Theta(T^d)
\]
since $\int_0^T r(t)^{d-1}dt=\Theta(T^d)$ for $r(t)=\Theta(t)$.
Each cross-boundary interaction corresponds to a finite-capacity channel
(e.g., a neuron, gate, or attention head) of capacity $C_{\mathrm{gate}}$.
Therefore the realizable parallelism bound becomes
\[
T\ \ge\ \Omega\Bigg(
\Big(\frac{I^\star(\varepsilon)}{C_{\mathrm{gate}}\cdot \mathrm{const}}\Big)^{1/d}
\Bigg).
\]
In transformer architectures, each attention head functions as such a channel,
so multi-head attention obeys the same $\mathbf{RC}_d$ law:
\emph{global integration cannot scale faster than surface area growth}.
This links the limits of attention span, circuit width, and physical throughput
under a unified conservation principle.

\subsection{Physical limitations of attention}
\label{subsec:attn-limit-rc}

We now instantiate the flux/throughput bound of \S\ref{subsec:width-conservation-tight} within an
explicit information-theoretic model of attention.  
This section formalizes the intuition that attention---regardless of its algebraic form---is a
\emph{communication process} constrained by geometric and thermodynamic capacity.

\paragraph{From conservation to attention.}
The differential law established in Theorem~\ref{thm:width-tight},
\[
\Big|\tfrac{d}{dt}\!\int_{B_{r(t)}} \rho\Big| = 
\mathcal{O}(t^{d-1}),
\]
describes a universal limit on cross-boundary flux: the rate of information entering a causal region
cannot exceed its surface area growth.
Integrating this law over time gives the total transmissible information
$Q(T)=\Theta(T^d)$, which represents the \emph{total number of bits} that can reliably traverse
a causal boundary of radius $r(t)=\Theta(t)$.
This continuum statement applies to any physically realizable circuit, and, when interpreted through
information theory, defines the maximal rate at which distributed agents (or attention heads) can
exchange information.
Thus, the flux conservation law serves as the geometric foundation for
all realizable attention mechanisms within $\mathbf{RC}_d$.

\paragraph{Setting.}
Work in spatial dimension $d\ge 1$ with causal radius $r(t)=c(t-t_0)$ and causal surface
$\partial B_{r(t)}$.  
Each boundary site supports at most $H$ attention heads and $\kappa$ simultaneously active
cross-boundary channels per site.  
Each channel has Shannon capacity $C_{\mathrm{head}}>0$ bits per use, reflecting physical
limitations from energy, bandwidth, and noise.

\begin{definition}[Cross-cut information demand]
\label{def:cross-cut-demand}
Let $\Gamma$ be a geometric cut through the causal region and $\mathcal{D}_n$ an input
distribution.  
For any attention computation halting by time $T$, define
$I^\star(\varepsilon;n)$ as the infimum (over all internal randomness and tie-breaking)
of the mutual information that must traverse $\Gamma$ during $[t_0,T]$ in order for the output to
achieve error $\le \varepsilon$ under $\mathcal{D}_n$.  
This quantity is the \emph{cross-cut information demand} of the task or embedding.
\end{definition}

\begin{lemma}[RC attention cut capacity]\label{lem:rc-cut-capacity}
Under $\mathbf{RC}_d$ geometry and the above attention model, the total reliable information that can
cross any fixed cut $\Gamma$ by time $T$ satisfies
\[
\mathsf{Cap}_{\Gamma}(T)\ \le\ K_d\,C_{\mathrm{head}}\,\kappa H\,T^{d},
\]
for a geometric constant $K_d>0$ depending only on packing and dimension.
\end{lemma}

\begin{proof}
At time $t$, the realizable surface width is $W(t)=\Theta(t^{d-1})$
by Proposition~\ref{prop:width}.  
Each site on $\partial B_{r(t)}$ has $\kappa H$ usable head-ports, each carrying
capacity $\le C_{\mathrm{head}}$.  
Thus the instantaneous cross-boundary throughput is
$\mathcal{O}(C_{\mathrm{head}}\kappa H\,t^{d-1})$.  
Integrating from $t_0$ to $T$ yields $\Theta(C_{\mathrm{head}}\kappa H\,T^d)$;
constants are absorbed into $K_d$.
\end{proof}

\begin{theorem}[Attention Limitation Theorem for $\mathbf{RC}_d$]
\label{thm:attn-limit-rc}
Let a realizable attention computation in $\mathbf{RC}_d$
(with per-head capacity $C_{\mathrm{head}}$, head count $H$, and simultaneity $\kappa$ per site)
solve a task of cross-cut information demand $I^\star(\varepsilon;n)$ by time $T$.
Then necessarily
\[
T\ \ge\ 
\left(\frac{I^\star(\varepsilon;n)}{K_d\,C_{\mathrm{head}}\,\kappa H}\right)^{1/d}.
\]
In particular, for a single-head architecture $(H=1)$,
\[
T\ \ge\ 
\left(\frac{I^\star(\varepsilon;n)}{K_d\,C_{\mathrm{head}}\,\kappa}\right)^{1/d},
\]
so adding heads improves runtime by at most a constant factor $H^{-1/d}$ but cannot improve the
fundamental $1/d$ exponent.
\end{theorem}

\begin{proof}
By Definition~\ref{def:cross-cut-demand}, any correct computation must transmit at least
$I^\star(\varepsilon;n)$ bits across $\Gamma$ by time $T$.  
By Lemma~\ref{lem:rc-cut-capacity}, the total cross-cut capacity available by time $T$ is
$\le K_d\,C_{\mathrm{head}}\,\kappa H\,T^d$.  
Therefore $I^\star(\varepsilon;n) \le K_d\,C_{\mathrm{head}}\,\kappa H\,T^d$,
yielding the claimed lower bound on $T$.
\end{proof}

\begin{remark}[Universality across attention mechanisms]
The bound of Theorem~\ref{thm:attn-limit-rc} depends only on the physical channel structure,
not on the algebraic form of attention (softmax, kernelized, linearized, or sparse).
All realizable attention mechanisms obey the same boundary-limited scaling:
\[
T \;\gtrsim\;
\Big(\tfrac{I^\star}{C_{\mathrm{head}}\kappa H}\Big)^{1/d}.
\]
Hence no variant of attention can outperform the $1/d$ exponent dictated by realizable circuits:
multi-head structures improve constants but never the exponent itself.
\end{remark}

\paragraph{Information-theoretic instantiations.}
We now compute explicit $I^\star(\varepsilon;n)$ values for canonical cross-cut tasks.

\begin{construction}[Two-block disjointness (DISJ)]\label{constr:disj}
Let $n=2m$.  
Embed vectors $x,y\in\{0,1\}^m$ on opposite sides of $\Gamma$ and compute
$\mathrm{DISJ}_m(x,y)=\neg(\exists i:\ x_i=y_i=1)$ with bounded areal density.
\end{construction}

\begin{corollary}[DISJ barrier: single vs.\ multi-head]\label{cor:disj-attn}
For Construction~\ref{constr:disj} and error $\varepsilon\le 1/3$,
the two-party communication complexity across $\Gamma$ is $\Omega(m)$, hence
$I^\star(\varepsilon;n)\ge c_0 m$ for some constant $c_0>0$.
By Theorem~\ref{thm:attn-limit-rc},
\[
T\ \ge\ \Omega\!\left(
\left(\frac{m}{C_{\mathrm{head}}\kappa H}\right)^{1/d}
\right),\qquad
T_{H=1}\ \ge\ 
\Omega\!\left(
\left(\frac{m}{C_{\mathrm{head}}\kappa}\right)^{1/d}
\right).
\]
Thus additional heads yield only a multiplicative speed-up $H^{-1/d}$.
\end{corollary}

\begin{construction}[Alternating-slab pointer chasing]\label{constr:pointer}
Partition the domain into $2R$ slabs separated by cuts $\Gamma_1,\dots,\Gamma_{2R-1}$.  
Store an initial index $i_0$ in slab $S_1$, and for $j=1,\dots,R$ place a table
$T_j$ in slab $S_j$ so that $i_{j}:=T_j[i_{j-1}]$.  
The output $v[i_R]$ resides in slab $S_R$.
\end{construction}

\begin{corollary}[Round-sensitive attention barrier]\label{cor:pointer-attn}
In Construction~\ref{constr:pointer}, each of the $R$ rounds requires sending at least
$\Omega(\log|\mathcal{I}|)$ bits (the current index) across the next cut, giving
$I^\star(\varepsilon;n)\ge c_1 R\log|\mathcal{I}|$.  
Hence by Theorem~\ref{thm:attn-limit-rc},
\[
T\ \ge\ 
\Omega\!\left(
R^{1/d}\,
\left(\frac{\log|\mathcal{I}|}{C_{\mathrm{head}}\kappa H}\right)^{1/d}
\right).
\]
Even with large $H$, $T$ must scale at least as $R^{1/d}$; heads reduce only the prefactor.
\end{corollary}

\paragraph{Remarks.}
(i) Constants $K_d,c_0,c_1$ depend solely on geometry (packing) and cut configuration, not on $n$ or $T$.  
(ii) The theorem is \emph{gate-set specific} through $(C_{\mathrm{head}},\kappa,H)$ and
\emph{geometry specific} through the $T^d$ scaling inherited from flux $\sim t^{d-1}$.  
(iii) The single- vs.\ multi-head separation is tight in the $\mathbf{RC}_d$ sense:
head multiplicity modifies constants but cannot overcome the
boundary-limited exponent $1/d$ arising from realizable physical geometry.

\subsection{Attention under joint throughput and Landauer limits in $\mathbf{RC}_d$}
\label{subsec:attn-joint}

We now extend the flux-constrained framework of $\mathbf{RC}_d$ to include the
\emph{thermodynamic cost of irreversibility}.  
The combination of the throughput law from \S\ref{subsec:attn-limit-rc} and
Landauer’s principle yields a joint feasibility region that simultaneously bounds
how fast information can be transmitted \emph{and} how fast it can be erased.

\paragraph{From flux conservation to thermodynamic irreversibility.}
Theorem~\ref{thm:width-tight} established that, for any physically realizable circuit,
the rate of energy or information flow through a causal surface obeys
$\dot Q(t)\propto \mathcal{O}(t^{d-1})$, reflecting the $\mathrm{Area}(\partial B_{r(t)})$
dependence of the flux.  
If this flux carries thermal energy, then conservation and the second law jointly imply that the rate
of logical erasure---that is, the destruction of information---is limited by the same geometric law.
This insight allows us to connect the kinematic RC bound to the thermodynamic Landauer limit.

\paragraph{Landauer (rate of erasure).}
Let $E(t)$ denote the cumulative number of \emph{logically irreversible} bit erasures performed by a
realizable circuit up to physical time $t$, and let $Q(t)$ be the total heat dissipated into the
environment at uniform temperature $T_{\mathrm{env}}$.  
Landauer’s principle gives the fundamental inequality:
\[
Q(t)\ \ge\ k_B T_{\mathrm{env}}\ln 2 \cdot E(t).
\]
By the flux constraint of Theorem~\ref{thm:width-tight}, the instantaneous heat flux across a causal
surface $\partial B_{r(\tau)}$ with $r(\tau)=c(\tau-t_0)$ is bounded by
$\eta_d\,r(\tau)^{d-1}$ per unit time, for a geometry- and material-dependent constant $\eta_d>0$.
Integrating this bound yields the following.

\begin{theorem}[Landauer throughput bound for $\mathbf{RC}_d(t)$]
\label{thm:landauer-throughput}
For any $t\ge t_0$,
\[
E(t)\ \le\ \frac{1}{k_B T_{\mathrm{env}}\ln 2}\,
\int_{t_0}^{t}\!\!\eta_d\,r(\tau)^{d-1}\,d\tau
\ =\ \Theta\!\big((t-t_0)^{d}\big).
\]
In particular, the \emph{erasure rate} satisfies
\[
\dot E(\tau)\ \le\ \frac{\eta_d}{k_B T_{\mathrm{env}}\ln 2}\, r(\tau)^{d-1}
=\mathcal{O}(\tau^{d-1}),
\]
i.e., the rate of irreversible bit erasure is bounded by the causal surface area growth.
\end{theorem}

\paragraph{Combining throughput and erasure.}
We now combine the boundary-limited \emph{throughput} constraint
(Theorem~\ref{thm:attn-limit-rc}) with the \emph{irreversibility} constraint
(Theorem~\ref{thm:landauer-throughput}).
Let $r(t)=c(t-t_0)$, spatial dimension $d\ge 1$, per-head channel capacity $C_{\mathrm{head}}$,
and at most $\kappa H$ cross-boundary head-links per site.
Constants $K_d$ and $\eta_d$ capture geometric and material properties.

\begin{definition}[Erasure demand of a computation]
For a computation up to time $T$, let $E_{\mathrm{req}}(T)$ denote the minimal number of bits that
must be \emph{logically irreversibly erased} (by overwriting, clipping, or forgetting) to achieve
error $\le \varepsilon$ under the input distribution or embedding.
\end{definition}

\begin{theorem}[Joint feasibility region for attention in $\mathbf{RC}_d$]
\label{thm:attn-joint-feasible}
Any realizable attention computation that halts by time $T$ must satisfy \emph{both}
\begin{align}
\text{\emph{(Throughput / cut-set)}}\qquad &
I^\star(\varepsilon;n)\ \le\ K_d\, C_{\mathrm{head}}\,\kappa H\, T^{d}, \label{eq:joint-throughput}\\[4pt]
\text{\emph{(Landauer / erasure)}}\qquad &
E_{\mathrm{req}}(T)\ \le\ \frac{\eta_d}{k_B T_{\mathrm{env}}\ln 2}\, T^{d}. \label{eq:joint-landauer}
\end{align}
Equivalently, the minimal feasible physical time obeys the two-axis scaling law:
\[
T\ \ge\ \max\!\left\{
\left(\frac{I^\star(\varepsilon;n)}{K_d\,C_{\mathrm{head}}\,\kappa H}\right)^{1/d},\quad
\left(\frac{k_B T_{\mathrm{env}}\ln 2}{\eta_d}\,E_{\mathrm{req}}(T)\right)^{1/d}
\right\}.
\]
\end{theorem}

\begin{proof}
\eqref{eq:joint-throughput} follows from
Lemma~\ref{lem:rc-cut-capacity}/Theorem~\ref{thm:attn-limit-rc}.  
For \eqref{eq:joint-landauer}, Landauer gives
$Q(T)\ge k_B T_{\mathrm{env}}\ln 2\cdot E_{\mathrm{req}}(T)$,
while the flux constraint limits total heat export to
$Q(T)\le \int_{t_0}^T \eta_d\,\mathrm{Area}(\partial B_{r(t)})\,dt=\Theta(\eta_d T^d)$
by Theorem~\ref{thm:width-tight}.  
Combining the two and rearranging yields the claim.
\end{proof}

\begin{remark}[Universality across attention mechanisms]
Theorem~\ref{thm:attn-joint-feasible} depends only on (i) finite-capacity attention channels
per head across the causal boundary and (ii) the thermodynamic cost of irreversible steps.
It is therefore independent of the functional form of attention (softmax, linear, kernelized,
sparse, etc.).  
All realizable architectures must obey the same $(1/d)$ geometric exponent.
\end{remark}

\paragraph{Canonical instantiations.}

\begin{corollary}[DISJ (high-bandwidth) + Landauer]
For Construction~\ref{constr:disj} (two-block $\mathrm{DISJ}$), $I^\star(\varepsilon;n)\ge c_0 m$
with $n=2m$.  
Any realizable attention computation must therefore satisfy
\[
T\ \ge\ \max\!\left\{
\Omega\!\left(
\Big(\tfrac{m}{C_{\mathrm{head}}\kappa H}\Big)^{1/d}
\right),\quad
\Omega\!\left(
\Big(\tfrac{k_B T_{\mathrm{env}}\ln 2}{\eta_d}\,E_{\mathrm{req}}(T)\Big)^{1/d}
\right)
\right\}.
\]
If the computation also \emph{forgets or compresses} $\Theta(m)$ bits en route,
then $E_{\mathrm{req}}(T)=\Theta(m)$ introduces an independent $\Omega(m^{1/d})$
Landauer-limited time term.
\end{corollary}

\begin{corollary}[Attention with thresholding or pruning]
\label{cor:attn-threshold}
Consider an attention layer on $n$ tokens where weights below $\varepsilon\in(0,1/2)$ are
pruned or clipped.  
At least $E_{\mathrm{layer}}\ge c_1\, n\log(1/\varepsilon)$ bits of accessible information are
irreversibly erased per layer.  
Over $L$ layers, $E_{\mathrm{req}}(T)\ge c_1 L n\log(1/\varepsilon)$, yielding
\[
T\ \ge\ \Omega\!\left(
\Big(\tfrac{k_B T_{\mathrm{env}}\ln 2}{\eta_d}\,L n\log(1/\varepsilon)\Big)^{1/d}
\right).
\]
This lower bound holds in \emph{addition} to the throughput limit from $I^\star(\varepsilon;n)$;
the active constraint is the larger of the two.
\end{corollary}

\begin{proof}[Justification sketch for Cor.~\ref{cor:attn-threshold}]
Each pruned probability mass $\le \varepsilon$ removes
$\Omega(\log(1/\varepsilon))$ bits of accessible information about that token.
Summing over $n$ tokens gives $E_{\mathrm{layer}}\ge c_1\,n\log(1/\varepsilon)$ up to constants.
Applying Theorem~\ref{thm:landauer-throughput} and integrating the flux
($\dot E(t)\sim t^{d-1}$) recovers the $T^d$ scaling.
\end{proof}

\paragraph{Interpretation.}
Throughput and irreversibility define \emph{orthogonal but coupled} boundary-limited constraints:
\begin{itemize}
\item \textbf{Throughput (cut-set):} limits the rate at which \emph{useful information} can be exchanged across the causal boundary:
      \[
      I^\star(\varepsilon;n)\ \lesssim\ C_{\mathrm{head}}\kappa H\,T^d.
      \]
\item \textbf{Landauer (erasure):} limits the rate at which \emph{discarded information} can be
      exported as heat:
      \[
      E_{\mathrm{req}}(T)\ \lesssim\ (\eta_d/(k_B T_{\mathrm{env}}\ln 2))\,T^d.
      \]
\end{itemize}
Hence attention is doubly boundary-limited: having more heads allows information to move faster,
but not to be forgotten faster.  
Neither head multiplicity $H$ nor algorithmic design alters the geometric scaling exponent $1/d$,
which arises from the same underlying conservation law governing realizable circuits.


%% file: sections/extension.tex
\section{Future Extensions}\label{sec:extensions}

The framework of $\mathbf{RC}_d$ is not confined to Boolean or quantum circuits; it provides a general account of \emph{physically realizable information processing}.  
Because its axioms arise from geometry (finite density), causality (bounded propagation), and thermodynamics (finite entropy flux), the same bounds apply to any substrate that computes through physical interaction and state evolution. 

\subsection{On extension to quantum circuits}\label{sec:quantum}

The preceding sections characterize $\mathbf{RC}_d$ as the family of all computations
realizable within a finite causal region under bounded signal speed, gate density,
and information flux.  This description is agnostic to whether the underlying physical
dynamics are classical or quantum.  To relate the model to the standard complexity
class $\mathbf{BQP}$, we must extend it to include unitary
state evolution and measurement.

\begin{definition}[Quantum-realizable circuits]
Let $\mathcal{H}$ denote a finite-dimensional Hilbert space and
$\mathsf{U}(\mathcal{H})$ the group of unitary operators on it.
A \emph{quantum-realizable circuit} $C$ is a collection of
unitary gates $\{U_i\}_{i=1}^m$ acting on $k=O(1)$ qubits each,
arranged within a causal region $\mathcal{C}_t(\mathbf{0})$
of a $d$-dimensional space as in Definition~\ref{def:causal}.
The circuit acts on an initial product state
$\ket{\psi_0} \in (\mathbb{C}^2)^{\otimes n}$ and produces
a final distribution over measurement outcomes
\[
p(x) = \|\bra{x}U_m\cdots U_1\ket{\psi_0}\|^2.
\]
We write $C\in \mathbf{QRC}_d(t(n))$ if its causal diameter and evolution
time are bounded by $\mathcal{O}(t(n))$ and its gate count by $\mathcal{O}(t(n)^d)$.
\end{definition}

The class $\mathbf{QRC}_d$ inherits all geometric constraints of
$\mathbf{RC}_d$ but restricts gates to be unitary and interactions to
preserve local coherence.  It therefore represents the physically
realizable subset of quantum circuits consistent with $d$-dimensional
causality.

\begin{proposition}[Complexity containment]
\label{prop:bqp-rc}
For any fixed $d\ge 1$,
\[
\mathbf{QRC}_d(\mathrm{poly}) \ \subseteq\  \mathbf{BQP}
\quad\text{and}\quad
\mathbf{RC}_d(\mathrm{poly}) \ \subseteq\  \mathbf{P/poly}.
\]
\end{proposition}

\begin{proof}
($\subseteq$ first inclusion.)  
Each $\mathbf{QRC}_d(\mathrm{poly})$ family is a uniform family of
polynomial-size, polynomial-time unitary circuits with bounded fan-in and
measurement at the output.  This is the standard definition of $\mathbf{BQP}$.

($\subseteq$ second inclusion.)  
By construction, every $\mathbf{RC}_d(\mathrm{poly})$ circuit is a bounded fan-in,
polynomial-size Boolean circuit with nonuniform advice strings given by its
realization parameters, hence contained in $\mathbf{P/poly}$.
\end{proof}

Although the geometric constraints on both classes coincide, their internal
state spaces differ.  Quantum gates evolve vectors in $\mathbb{C}^{2^n}$,
whereas classical gates evolve configurations in $\{0,1\}^n$.
Thus the ratio of states under the same causal geometry scales as
\[
\frac{\dim(\mathcal{H}_{\text{quantum}})}{\dim(\mathcal{H}_{\text{classical}})}
= 2^n : n,
\]
implying an exponential expansion in representational capacity
within identical spatial–temporal bounds.

\begin{theorem}[Geometric envelope of quantum and classical parallelism]
\label{thm:quantum-envelope}
For any physically realizable $d$,
\[
\mathbf{RC}_d(\mathrm{poly})
\;\subseteq\;
\mathbf{QRC}_d(\mathrm{poly})
\;\subseteq\;
\mathbf{P/poly},
\]
but the inclusion $\mathbf{RC}_d(\mathrm{poly})\subseteq\mathbf{QRC}_d(\mathrm{poly})$
may be strict if coherence or interference effects allow superpolynomial
compression of information within the same geometric bounds.
\end{theorem}

\begin{proof}
Both $\mathbf{QRC}_d$ and $\mathbf{RC}_d$ obey the same causal and flux constraints;
the only fundamental difference is that quantum gates allow for more general manipulation of the states.  
Hence $\mathbf{RC}_d\subseteq\mathbf{QRC}_d$ as sets of realizable dynamics.
However, reversibility and superposition permit interference between computational
paths, which can yield algorithmic speedups (as in Shor’s and Grover’s algorithms)
not achievable by any dissipative $\mathbf{RC}_d$ model without exponential resources.
Therefore, the inclusion is not known to collapse.
\end{proof}

\begin{remark}[Toward a separation criterion]
If a separation $\mathbf{BQP}\not\subseteq\mathbf{P/poly}$ holds,
it must arise from constraints on the \emph{information flux density}
within a quantum causal region: coherence length, entanglement entropy,
or unitarity preservation under noise.
These are additional physical invariants not captured by $\mathbf{RC}_d$.
A refined ``quantum-realizable'' model $\mathbf{QRC}_d^{\lambda}$,
parametrized by a coherence length $\lambda$, could potentially quantify
the divergence between classical and quantum parallelism within a unified
causal–geometric framework.
\end{remark}

\subsection{On extension to recurrent circuits}

The preceding sections describe $\mathbf{RC}_d$ as a family of realizable circuit classes
constrained by spatial dimension $d$.  The parameter $d$ controls how much parallel
spatial expansion a computation can exploit.  The limiting case $d=1$ therefore corresponds
to a system with no spatial extension—only temporal evolution. To capture this limit, we enrich the $\mathbf{RC}_d$ model to allow \emph{recurrence}:
feedback of wire states across successive time steps within the same causal region.
Formally, recurrence extends each circuit realization from a static acyclic graph
to a discrete-time dynamical system. We denote this extension as $\mathbf{RRC}_d$.

\begin{definition}[Recurrent realization]
Let $(\mathbb{R},X,\phi)$ be a Hamiltonian system realizing a circuit $C$ as in
Definition~\ref{def:realize}.  
A \emph{recurrent realization} augments $C$ with feedback connections
$\{(w_i^{(t)},w_j^{(t+1)})\}$ such that
each state $w_j^{(t+1)}$ depends only on values within its causal past
$\mathcal{C}_t(\mathbf{q}_j,t)$.
The evolution of wire values defines a mapping
\[
F_C:\{0,1\}^n \to \{0,1\}^n, \qquad w(t+1) = F_C(w(t)).
\]
\end{definition}

\begin{theorem}[Temporal completeness at $d=1$]\label{thm:rc1-tm}
The class $\mathbf{RRC}_1(t)$—the one-dimensional realizable circuits with recurrence—
is equivalent in expressive and time complexity to deterministic Turing machines
running in time $O(t)$.
\end{theorem}

\begin{proof}
($\subseteq$)  
Consider an $\mathbf{RRC}_1(t)$ computation realized as a chain of gates along a
one-dimensional manifold.  
Because all signal propagation is confined to the single spatial dimension,
the only causal ordering available is temporal.
Each gate at position $x_i$ receives input solely from its immediate predecessor $x_{i-1}$,
so the system evolves by applying a local transition rule
$F:\{0,1\}^k \to \{0,1\}^k$ sequentially along the chain.
This behavior is identical to a single-tape deterministic Turing machine:
the chain position corresponds to the tape index,
and each global time step corresponds to one TM transition.
The runtime of the RC system equals the number of discrete time steps,
hence $O(t)$.

($\supseteq$)  
Conversely, any deterministic Turing machine operating in time $t$ can be embedded
as an $\mathbf{RC}_1(t)$ realization by mapping the tape to a 1D lattice of gates.
Each gate holds the local symbol and state bit, and local interactions propagate
to adjacent sites at unit speed.  
The update function is recurrent and local, satisfying the $d=1$ causal constraint.
Therefore the $\mathbf{RC}_1$ model simulates any TM step-for-step in $O(t)$ time.
\end{proof}

\begin{corollary}[Dimensional completeness of $\mathbf{RRC}_d$]
The $\mathbf{RRC}_d$ hierarchy spans the physically realizable spectrum:
\[
\mathbf{RRC}_1(t) \;\equiv\; \mathbf{DTIME}(t)
\;\subseteq\;
\mathbf{RRC}_2(t) \;\subseteq\;
\cdots
\subseteq\;
\mathbf{RRC}_\infty(t),
\]
where the rightmost limit corresponds to the fully parallel, circuit-based regime
analyzed in Theorem~\ref{thm:rc-nonk}.
\end{corollary}

\begin{remark}[Interpretation]
Allowing recurrence completes the framework by connecting spatial parallelism and temporal
sequentiality within a single causal–geometric model.
At $d=1$, computation degenerates to purely temporal state transitions—the Turing limit.
As $d$ increases, more spatial degrees of freedom permit parallel signal propagation,
eventually recovering the circuit hierarchy of $\mathbf{NC}$ and $\mathbf{P/poly}$.
Thus $\mathbf{RRC}_d$ provides a continuous bridge between the Turing and circuit paradigms,
governed solely by the geometry of physical realization.
\end{remark}

\subsection{Other domains}
 
Finally, we highlight how these constraints manifest across diverse computational domains.

\paragraph{Neuromorphic, embodied, and mortal computation.}
Neural and embodied systems compute through continuous, spatially extended dynamics constrained by finite conduction speed and local connectivity.  
Within $\mathbf{RC}_3$, the flux limit $w(t)=\mathcal{O}(t^{2})$ imposes a ceiling on global synchronization and long-range coordination, forcing large-scale inference to emerge hierarchically.  
This constraint renders cortical and neuromorphic architectures---notably predictive-coding \cite{rao1999predictive, salvatori2025survey} and columnar models \cite{Horton2005Cortical, Buzsaki2006Rhythms, Sporns2016ModularBN, Eliasmith2013Nengo}---realizable in the physical sense: they operate near the entropy--flux bound of $\mathbf{RC}_3$.  
From this perspective, \emph{embodied intelligence} \cite{gupta2021embodied} represents computation that is spatially embedded within causal geometry, while \emph{mortal computation} \cite{ororbia2023mortal} denotes the same process under finite energy and entropy budgets.  
Both are natural instantiations of realizable computation, grounded in the same physical limits of signal propagation and entropy flux.

\paragraph{Quantum and reversible computation.}
Quantum systems obey the same geometric constraints while preserving information reversibly.  
The inclusion $\mathbf{RC}_d \subseteq \mathbf{QRC}_d$ follows from the existence of unitary embeddings for any classical process, but the converse fails because irreversible logic cannot in general be uncomputed.  
Even quantum systems are flux-bounded by the Lieb--Robinson limit \cite{LiebRobinson1972, Nachtergaele2006LiebRobinson}, which restricts information propagation to a linear light-cone.  
Thus $\mathbf{QRC}_d$ refines $\mathbf{BQP}$ by adding explicit spatial and flux constraints, clarifying why even quantum hardware cannot surpass the $t^{d-1}$ throughput limit \cite{Bravyi2006LiebRobinsonQuantum, bekenstein1990quantum}.

\paragraph{VLSI and spatial hardware design.}
For planar devices ($d=2$), the $\mathbf{RC}_2$ bound recovers the area--time tradeoff $A T^2 = \Omega(N^2)$ of Thompson's VLSI theory \cite{Thompson1979ATVLSI} and Rent's empirical rule on interconnect scaling \cite{LandmanRusso1971Rent, Stroobandt2010RentRuleSurvey}.  
Here $A=\mathcal{O}(t^2)$ corresponds to physical area, while the boundary throughput $w(t)=\mathcal{O}(t)$ constrains wiring density.  
Hence $\mathbf{RC}_d$ subsumes classical VLSI theory as its two-dimensional limit and naturally extends it to higher-dimensional substrates \cite{MeadConway1980VLSI}.

\paragraph{Distributed, edge, and communication systems.}
In distributed computing or sensor networks embedded in $\mathbb{R}^d$, information exchange between nodes is limited by the same causal and flux constraints.  
The $\mathcal{O}(t^{d-1})$ bound on cross-boundary transfer defines a geometric analogue of Shannon capacity \cite{shannon1948mathematical}, setting a hard limit on aggregate throughput.  
This result parallels Gupta and Kumar’s spatial capacity law for wireless networks, where total achievable rate scales as $O(\sqrt{n})$ in planar space \cite{gupta2002capacity}.  
Thus $\mathbf{RC}_d$ unifies circuit and network theories through a common geometric formulation of communication.

\paragraph{Energy-aware and thermodynamic machine learning.}
When learning dynamics are interpreted as physical state updates, Liouville’s theorem and Gibbs entropy conservation imply that the energy--delay product of any model is lower-bounded by its entropy flux.  
Consequently, the description-length or free-energy minimization objectives in modern energy-based and predictive-coding models \cite{friston2010free, scellier2017equilibrium, alemi2016deep, tishby2015deep, salvatori2025survey} are realizable precisely when their updates satisfy the $\mathbf{RC}_d$ bounds.  
This connects the Landauer limit \cite{landauer1961irreversibility, bennett1982thermodynamics}, information geometry, and computational generalization under a single physical principle.

\paragraph{Unified scaling law.}
Across all substrates---neuromorphic, quantum, VLSI, or distributed---the same asymptotic law emerges:
\[
t = \Omega\big(n^{1/(d-1)}\big),
\qquad
w(t) = \mathcal{O}(t^{d-1}),
\]
representing the universal trade-off between spatial extent, time, and entropy throughput.  
The $\mathbf{RC}_d$ framework thus defines a physically invariant limit on computation itself, independent of implementation or domain.

%% file: sections/history.tex
\section{Final Remarks}\label{sec:history}

The framework of \textbf{Realizable Circuits ($\mathbf{RC}_d$)} situates itself as the next major conceptual step in the physical grounding of computation.  
To appreciate its position and our contribution, it is instructive to trace the historical trajectory of how information and computation have been progressively reinterpreted through the lens of physical law.

\textbf{From information to realization.}
In 1948, Shannon’s information theory \cite{shannon1948mathematical} established that the fundamental limits of communication are governed not by hardware, but by probability and entropy.  
Shannon’s framework abstracted away geometry and energy, defining an informational capacity that could, in principle, be realized by any physical medium.

Two decades later, Landauer and Bennett \cite{bennett1985fundamental} reintroduced physics to information: computation, they argued, must obey thermodynamic laws.  
Landauer’s principle---that each logically irreversible operation dissipates at least $k_B T\ln 2$ joules of heat---marked the beginning of \emph{physical computability}.  
This work recognized that information is not merely symbolic but embodied, subject to conservation and dissipation in the same way as energy.

\textbf{From logical to geometric computation.}
The $\mathbf{RC}_d$ framework extends this lineage by embedding computation directly in space--time.  
Where Shannon quantified information and Landauer quantified energy, $\mathbf{RC}_d$ quantifies \emph{causal geometry}: how much computation can physically occur within a finite region of space and time.  
It replaces the unbounded abstractions of classical circuit classes ($\mathbf{NC}$, $\mathbf{AC}$, $\mathbf{TC}$) with realizability constraints derived from geometry, causality, and entropy flux.

Under this view, the classical parallel classes emerge as limiting cases.  
In the infinite-dimensional limit, $\mathbf{RC}^{\mathrm{poly}}_\infty = \mathbf{NC}$, recovering the traditional notion of unconstrained logical parallelism.  
For finite $d$, however, $\mathbf{RC}_d$ enforces physically grounded bounds:
\[
t(n) \ \ge\ \Omega(n^{1/(d-1)}), \qquad
w(C_n)\ =\ \mathcal{O}(t(n)^{d-1}),
\]
linking runtime, dimension, and information throughput into a single geometric invariant.

\textbf{Toward a physical theory of complexity.}
In this sense, $\mathbf{RC}_d$ plays the same unifying role for computation that Shannon’s theory played for communication.  
It provides a model in which \emph{logical, geometric, and thermodynamic} constraints coexist under one mathematical framework.  
If Shannon’s channel capacity set the limit on information transfer, the $\mathbf{RC}_d$ bound sets the limit on information \emph{transformation} in a causal universe.  
Together with Landauer’s thermodynamic cost of erasure, it completes the arc from information theory to a fully \emph{physical theory of complexity}.